\documentclass[%
 reprint,
nofootinbib,
 amsmath,amssymb,
 aps,
]{revtex4-2}
\usepackage{graphicx}
\usepackage{dcolumn}
\usepackage{bm}

\usepackage{latexsym,verbatim}
\usepackage{epsfig,epstopdf,color}
\usepackage{float}
\usepackage{soul}
\usepackage{physics}
\usepackage{esint}
\usepackage{empheq}
\usepackage[normalem]{ulem}
\usepackage{xcolor}
\usepackage{amsthm}
\usepackage{bbm}

\usepackage[colorlinks=true, allcolors=blue]{hyperref}

\usepackage{listings}
\newtheorem{theorem}{Theorem}

\newcommand{\JR}[1]{\textcolor{blue}{JR: #1}}
\newcommand{\DC}[1]{\textcolor{orange}{DC: #1}}

\begin{document}

\title{Radial Mode Stability of Two-Fluid Neutron Stars}
\author{Daniel A. Caballero}
\affiliation{Illinois Center for Advanced Studies of the Universe \& Department of Physics,
University of Illinois at Urbana-Champaign, Urbana, Illinois 61801, USA.}

\author{Justin Ripley}
\affiliation{Illinois Center for Advanced Studies of the Universe \& Department of Physics,
University of Illinois at Urbana-Champaign, Urbana, Illinois 61801, USA.}

\author{Nicol\'as Yunes}
\affiliation{Illinois Center for Advanced Studies of the Universe \& Department of Physics,
University of Illinois at Urbana-Champaign, Urbana, Illinois 61801, USA.}


\begin{abstract}
    Radial mode stability is a necessary condition for the astrophysical viability of compact objects.
    In recent years, astrophysical models with two fluids have gained popularity, especially in their ability to model dark matter admixed neutron stars.
    Just as is the case of single-fluid stars, a stability criterion based on the background equations has been developed 
    --the critical curve for the particle numbers of the two fluids in the two-dimensional configuration space determines a one-dimensional sequence that labels the marginally stable configurations--
    but its validity depends on the linear stability of radial perturbations which remains unstudied.
    In this paper, we establish a set of stability criteria for two perfect-fluid relativistic stars by carefully studying the radial mode perturbation equations.
    We prove that modes are complete, have real eigenvalues with a minimum eigenvalue (i.e. a fundamental mode), thus a configuration is stable if and only if the fundamental mode is positive.
    As a consequence, our work formally and rigorously proves these necessary conditions for the stability criterion based on the background equations.
\end{abstract}

\maketitle

\section{Introduction}


Radial mode stability is a necessary condition for the astrophysical viability of compact objects, because, if unstable, these objects would either collapse to a black hole or explode in a supernova. 
For relativistic stellar models that consist of a single perfect fluid, there is a direct relation between the frequencies of radial perturbations and certain ``extremal points.''
In particular, a configuration will have a mode with zero eigenfrequency if and only if the equilibrium configuration corresponds to an extremum in the stellar mass with respect to the stellar radius \cite{BardeenThorneMeltzer1966,HTWW1965}, i.e.~an extremum of the ``mass-radius curve.''
The radial modes of single-fluid stars\footnote{In this paper, we will analyze perfect fluid stars only, and thus whenever we refer to ``single-fluid'' or ``multi-fluid'' stars, it is assumed that the fluids are perfect.} are real, complete and bounded from below, and thus, a configuration is radially mode stable if and only if its smallest eigenvalue is positive \cite{Chandrasekhar1964,HTWW1965}. 
These two results, coupled with the fact that low-density (non-relativistic) single-fluid solutions are radially stable \cite{BardeenThorneMeltzer1966,HTWW1965}, imply that the critical points of the mass-radius curve for single-fluid stars demarcate which solution regions in the mass-radius plane are radially stable.
As the mass-radius curve of the star is determined by the background solution, this relation allows one to determine  the \emph{perturbative} stability of a star purely by repeatedly solving the \emph{background} equilibrium equations, a much simpler task than solving the perturbed equations.

The single-fluid model is widely used due to its relative simplicity, and due to the general expectation that relativistic stars (such as white dwarfs and neutron stars) are composed of essentially a single fluid.
Nevertheless, stellar models that consist of multiple fluids, which we shall refer to as ``multi-fluid stars,'' may have physical applications. 
One such application is dark-matter-admixed neutron stars, compact objects modeled through two independently-conserved fluids that interact solely through gravitation \cite{Deliyergiyev2019,Rafiei2022,bell_improved_2020,Emma2022,DasKumarPatra2021}.
Dark-matter-admixed neutron stars are similar to ``normal'' baryonic stars, but could potentially differ in measurable ways that would inform our knowledge of dark matter \cite{Gresham2019,Lee_2021,CIARCELLUTI2011,Rutherford2023,Mariani2024,Sandin2009,Hippert2023}. 
Another application of multi-fluid stars with non-negligible interactions is the study of superfluid neutron stars \cite{AnderssonComer2021,Andersson2013,Comer2004,Comer1999,AnderssonComerGrosart2004} and
stars in the context of $f(R)$ gravity \cite{Campbell2024} and with a second gravitational sector in $f(\mathcal T)$ gravity \cite{Pradhan2024}.
Multi-fluid stars have also been studied in the context of the $1+1+2$ covariant formalism, where a theorem for generating analytical solutions was developed in \cite{Naidu2021}.

In spite of this interest, the conditions for stability of multiple-fluid stars have remained mostly understudied.
Reference \cite{HENRIQUES1990} extended the stability criterion of the maximum mass to non-interacting multi-fluid stars by invoking that stable configurations must satisfy particle number conservation for each fluid independently.
For a two-fluid star, the configuration space can be parameterized by the central densities of the two fluids, and thus, it is two-dimensional. 
The marginally stable configurations therefore form a curve in this two-dimensional space, defined as the extrema of \textit{both} the mass and the particle numbers with respect to the central densities.
This criterion has been employed in multiple studies of dark-matter-admixed neutron stars \cite{DiGiovanni2022,GOLDMAN2013,Hippert2023,Barbat2024},
and in the context of fermion-boson stars, where bosons are modeled as a scalar field \cite{HENRIQUES1990b,Valdez-Alvarado2020,Pitz2024}.

However, just as is the case of single-fluid stars, the ability to declare stability to radial perturbations using the above criterion hinges on the direct relation between critical points and radial perturbations.
Reference \cite{HENRIQUES1990} did not carry out a radial mode perturbative analysis of multi-fluid stars, and therefore, such an analysis is critical to establish whether their criterion does, in fact, determine radial mode stability.
Specifically, it is necessary to determine whether the radial modes of multiple perfect-fluid stars are real, complete and bounded from below.
Although equations for the modes have been obtained for both the non-interacting \cite{Kain2020} and interacting case \cite{Comer1999}, the analytic properties of the modes remain undetermined.
Moreover, very few numerical studies have been carried out to evaluate whether the criterion does in fact predict stability.
Reference \cite{Kain2021} is the only analysis to compare the criterion with a direct calculation of the radial modes,
while Ref.~\cite{Gleason2022,DiGiovanni2022,Valdez-Alvarado2013} simulated the dynamical evolution of two-fluid stars to determine stability.

A number of studies \cite{Mukhopadhyay2016,Gresham2019,Ellis2018,CIARCELLUTI2011,Lee_2021,DasKumarPatra2021,Rutherford2023} have used the condition $dM/dR=0$ to determine the stability of stellar solutions, but utmost care should be taken.
This approach--while applicable to single-fluid stars--cannot generally be applied to multi-fluid configurations,
because it fails to take into account the particle numbers and the radius associated with the secondary fluid, which are necessary in the criterion of \cite{HENRIQUES1990}.
Although using solely the mass and primary fluid radius should serve as a good approximation for stars with a low fraction of secondary fluid or dark matter,
the stability contour in \cite{Mukhopadhyay2016} (where just $M-R_{\rm SM}$ is used) and in \cite{Kain2021} (where the correct criterion is used and compared with eigenfrequencies) are qualitatively different.
The work of \cite{Deliyergiyev2019} correctly points out that computing eigenfrequencies is a safe criterion for stability over using mass-radius relations, but this paper uses the incorrect eigenvalue problem for studying stability (i.e.~the Newtonian single-fluid eigenvalue problem).
References \cite{Pontopoulos2017,JimenezFraga2022} computes the eigenvalues but for an eigenvalue problem that considers the two fluids together, which the authors in \cite{JimenezFraga2022} claim is an effective approximation to reduce computational resources.

In this paper, we establish a set of stability criteria for two-fluid relativistic stars by carefully studying the radial mode perturbation equations.
Specifically, we first prove that the radial modes of both fluids are real, complete and bounded from below, which we show is a necessary condition for the stability criterion  proposed in \cite{HENRIQUES1990}. Furthermore, we prove that a configuration will be stable if and only if its minimum eigenvalue is positive.
While obtaining these results, we also find a canonical energy that we split into a kinetic energy and a potential energy for the perturbations,
thus providing an alternative stability criterion (i.e.~positivity of canonical energy) in the same way as done for single-fluid stars \cite{HTWW1965,Prabhu_2016,Friedman1978,Shi2023} 

Our work, therefore, formally and rigorously establishes, for the first time, that the critical curve of the particle numbers of the two fluids in the two-dimensional configuration space determines a 1-dimensional sequence that labels the marginally stable configurations.
In other words, the 2 directional-derivative equations $dN_{1,2}/d{\boldsymbol{\rho_c}} = 0$, where ${\boldsymbol{\rho_c}} = (\rho_{c,1},\rho_{c,2})$ is the two-dimensional central density vector~\cite{HENRIQUES1990}, determines a 1-dimensional curve in the ${\boldsymbol{\rho_c}}$ configuration space above which stellar configurations are unstable to at least one radial perturbation mode.
In other words, for densities above those defined by this 1-dimensional curve, radial perturbations have imaginary frequencies, and thus, diverge exponentially.
This link between a critical curve and radial mode stability now allows future studies to bypass all radial perturbative analysis
and simply focus on finding the critical curve from the background solutions to determine the marginally stable configurations of two-fluid stars.   


Our paper is organized as follows.
In Sec.~\ref{sec:onefluid}, we review the basic formalism we will use and results for the radial mode stability of single-fluid stars.
In Sec.~\ref{sec:lang}, we consider the radial oscillations of two-fluid stars.
We use the general formalism established in \cite{Kain2020} (see also \cite{Comer1999}), 
and recast the equations into a symmetric form.
Using this, in Sec.~\ref{sec:energy}, we obtain a variational principle for the modes equivalent to the one for a single-fluid star.
We additionally find a canonical energy for the perturbations and a Lagrangian for the perturbative equations of motion.
From this, we derive a number of properties of radial perturbations, including that radial modes are real, complete and have frequencies that are bounded from below.
We additionally derive conditions for the radial modes to be stable (have all real frequencies).
In Sec.~\ref{sec:multi}, we outline how we can extend these results to stars made up of an arbitrary number of non-interacting fluids.
Finally, we conclude in Sec.~\ref{sec:conclusion}.
The Appendices contain more detailed derivations of several equations used in this article. Henceforth, we make use of the following conventions: Greek letters in indices stand for spacetime coordinates, while capital Latin letters stand for the fluid number, the metric signature is $(-,+,+,+)$, and we use geometric units where $G=1=c$.

\section{One Fluid Review}\label{sec:onefluid}

For the sake of being thorough and to set notation, we first review the stability of non-rotating, single fluid, relativistic stars as studied via the radial modes of pulsation.
While there are several approaches to studying stability \cite{Chandrasekhar1964,HTWW1965,BardeenThorneMeltzer1966}, we only review the direct perturbative approach taken in \cite{Chandrasekhar1964}, as we will later generalize that method to multi-fluid stars later.
We will also briefly review the energy argument used in Appendix B of \cite{HTWW1965},
and derive an explicit lower bound to the eigenvalues, which has not yet been explicitly done for relativistic stars.
Finally, we will connect all of this to determining the stability of an equilibrium configuration through the use of the mass-radius curve.

\subsection{TOV Equations}

Throughout this article, we analyze solutions to the Einstein Field equations coupled to the stress-energy tensor of a fluid, namely
\begin{align}  
    G_{\mu\nu} 
    &=
    8 \pi \, T_{\mu\nu} \ ,
    \label{eq:einstein-eqns}
\end{align}
The stress-energy tensor will be that of one (or many) perfect fluid(s), which is given by
\begin{equation}\label{eq:single-fluid-perfect}
    T^\mu{}_\nu=(\epsilon+p)u^\mu u_\nu+p\delta^\mu_\nu \ ,     
\end{equation}
where $u^\mu$ is a time-like vector, so that $u_\mu u^\mu=-1$, $\delta^\mu_\nu$ is the 4-dimensional Kronecker delta, and
$(p,\epsilon)$ are the pressure and the (total) energy density of the fluid respectively, which are related to each other through an equation of state $p=p(\epsilon)$. The stress-energy tensor is covariantly conserved, 
\begin{align}
    \nabla_{\mu}T^{\mu}{}_{\nu} =
    0 \,, \label{eq:conservation-law}
\end{align}
by the Einstein equations and the Bianchi identities. 

In order to study perturbations of neutron stars, we must first construct a \textit{background} spacetime upon which we can study perturbations. This background spacetime must be a solution to the Einstein equations that can represent a neutron star. We assume here that the neutron star is static and spherically symmetric, so that we may parameterize the metric via
\begin{equation}\label{eq:metric}
    ds^2=-e^{-2\Phi(r)}dt^2+e^{2\Lambda(r)}dr^2+r^2d\Omega^2 \ .
\end{equation}
where $\Phi(r)$ and $\Lambda(r)$ are free functions of the radius only. Using this background metric ansatz in the Einstein equations yields differential equations for the free functions  $\Phi(r)$ and $\Lambda(r)$.

A simple way to derive these differential equations is to project Eq.~\eqref{eq:conservation-law} onto $u^{\mu}$ and its perpendicular projection $\Pi^{\mu}{}_{\nu}\equiv u^{\mu}u_{\nu} + \delta^{\mu}_{\nu}$. Doing so, one obtains 
\begin{gather}
    u^\nu\nabla_\nu \epsilon+(\epsilon+p)\nabla_\nu u^\nu=0 \ , \label{eq:ener}\\
    (\epsilon+p)u^\nu \nabla_\nu u^\mu+(g^{\mu\nu}+u^\mu u^\nu)\nabla_\nu p=0 \ . \label{eq:euler}
\end{gather}
Equation~\eqref{eq:euler} is the relativistic Euler equation, and Eq.~\eqref{eq:ener} is the relativistic energy equation. Substituting the metric in Eq.~\eqref{eq:metric} into the Euler equation [Eq.~\eqref{eq:euler}], we obtain the relativistic equation of hydrostatic equilibrium
\begin{align}\label{eq:hydro}
    \frac{d}{dr}p=(\epsilon+p)\frac{d}{dr}\Phi \ .
\end{align}
To obtain equations for the metric components, we use the stress-energy tensor given in Eq.~\eqref{eq:single-fluid-perfect} as the source for Eq.~\eqref{eq:einstein-eqns}.
Only the $(t,t)$ and $(r,r)$ components are independent, from which we obtain
\begin{align}
    \label{eq:lam}
    \frac{d}{dr}\Lambda
    &=
    \frac r2\left[e^{2\Lambda}\left(8\pi\epsilon+\frac{1}{r^2}\right)-\frac{1}{r^2}\right] \ ,
    \\
    \label{eq:phi}
    \frac{d}{dr}\Phi
    &=
    \frac r2\left[-e^{2\Lambda}\left(8\pi p+\frac{1}{r^2}\right)+\frac{1}{r^2}\right] \ . 
\end{align}
Subtracting Eq.~\eqref{eq:lam} and Eq.~\eqref{eq:phi} results in 
\begin{align}\label{eq:diff}
    \frac{d}{dr}\Lambda-\frac{d}{dr}\Phi=4\pi re^{2\Lambda}(\epsilon+p) \ .
\end{align}

The above equations [Eqs.~\eqref{eq:hydro},~\eqref{eq:lam} and~\eqref{eq:phi}] form a closed system of differential equations for the fields $(p,\epsilon,\Phi,\Lambda)$ once an equation of state $p=p(\epsilon)$ is prescribed, but usually, one recasts these equations in a simpler way. 
Instead of working with the $\Lambda(r)$ metric function, the field $m(r)$ is defined through
\begin{equation}\label{eq:defm}
    1-\frac{2m}{r}=e^{-2\Lambda} \ ,
\end{equation}
so that Eqs.~\eqref{eq:lam} and \eqref{eq:phi} become
\begin{align}
    \frac{d}{dr}m
    &=
    4\pi r^2\epsilon \ , \label{eq:m}\\
    \frac{d}{dr}\Phi
    &=
    -\frac{1}{r^2}\frac{1}{1-\frac{2m}{r}}(4\pi p r^3+m) \ . \label{eq:phi2}
\end{align}
Combining the latter with the relativistic equation of hydrostatic equilibrium, we obtain
\begin{align}\label{eq:TOV}
    \frac{d}{dr}p=- \left(\frac{\epsilon+p}{r^2} \right) \left(\frac{4\pi p r^3+m}{1-\frac{2m}{r}}\right) \ ,
\end{align}
which is sometimes referred to as the Tolman-Oppenheimer-Volkoff (TOV) equation. In these new variables, Eqs.~\eqref{eq:m},~\eqref{eq:TOV} and the equation of state form a closed system for the fields $(p,\epsilon,m)$; once this system is solved, one can then solve Eq.~\eqref{eq:phi2} to find the remaining metric potential. 

The equation of state clearly plays an important role in determining the properties of the solution to these structure equations. We will refrain from selecting a particular equation of state in this paper, so that our results remain as general as possible. However, we will still find it convenient to define a few additional fields and variables that characterize the types of fluids we will be considering. First, we assume that we can assign a particle number density $n$ to the relativistic fluids we study in this article, and that this particle number density is conserved,
\begin{equation}\label{eq:consv}
    \nabla_\nu(n u^\nu)=0 \ .
\end{equation}
The conservation of particle number implies we are assuming the fluid is in equilibrium, and that there are no nuclear reactions that induce a net outward flux of particles in play.    
We can always relate $n$ to the other fluid variables $(\epsilon,p)$ via the first-law of thermodynamics for an isentropic star, 
\begin{equation}\label{eq:first}
    \frac{d\epsilon}{dn}=\frac{\epsilon+p}{n} \ .
\end{equation}
We note that Eqs.~\eqref{eq:ener}, \eqref{eq:consv} and \eqref{eq:first} are consistent with one another, such that only two of the three equations are independent. 
We next introduce the adiabatic index
\begin{equation}\label{eq:gamma}
    \gamma=\frac{d\log p}{d\log n} \ ,
\end{equation}
and the speed of sound
\begin{equation}
    c_s^2=\frac{d p}{d\epsilon} \ .
\end{equation}
The two quantities are related by the first law of thermodynamics [Eq.~\eqref{eq:first}] via
\begin{equation}\label{eq:gcs}
    \gamma \, p=(\epsilon+p)c_s^2 \ .
\end{equation}
Note here that just because we have defined an adiabatic index does not mean we are considering polytropic equations of state. When considering such equations of state, the adiabatic index is a constant, where for us $\gamma$ above need not be.  

Many numerical methods exist to solve the system of equations described above. In a slight abuse of terminology, we will refer to the system of Eqs.~\eqref{eq:hydro}, \eqref{eq:m} and \eqref{eq:phi2} (and later on their generalization to two fluids) as the TOV equations\footnote{We will also refer to them as TOV equations, even when we use $\Lambda$ instead of $m$ for the field variable.}.
One way to numerically integrate the TOV equations is to specify $m(0)=0$ and $p(0)=p_c$ at the center of the star\footnote{Numerically, we cannot start the integrations at the center of the star, so one usually starts at a small enough radius at which we can find an asymptotic solution to the TOV equations to initialize the integrations.}.
The value of $p_c$ is usually determined by giving a value for the central density $\epsilon_c=\epsilon(0)$ and using the EoS, such that $p_c=p(\epsilon_c)$ and the initial condition is completely specified by $\epsilon_c$.
The equations are integrated up to a radius coordinate $R$ at which\footnote{Numerically, we cannot integrate the TOV equations until the pressure exactly vanishes, so one usually defines the stellar radius as the radial coordinate at which the pressure is many orders of magnitude smaller than the central pressure.} $p(R)=0$. 
Then, $R$ is the radius of the star, and its mass is $M\equiv m(R)$.
From the solution to the TOV equations, we obtain $p(r)$ and $m(r)$, from which we obtain $\epsilon(r)$ by inverting the EoS $\epsilon(r)=p_{EoS}^{-1}(p(r))$ and the metric function $\Lambda(r)$ by inverting Eq.~\eqref{eq:defm}.
We then obtain $\Phi(r)$ by integrating Eq.~\eqref{eq:phi2} from the radius of the star to the center, subject to the boundary condition $\Phi(R)=-\frac12\log\left[1-{2m(R)}/{R}\right]$. 
With this in hand, we can obtain the particle density $n$ as a function of the radius by integrating Eq.~\eqref{eq:first} from the radius $R$ to the origin, under the boundary condition $n(R)=0$.
If needed, we can also obtain the total number of particles in the star, $N$, by integrating $n(r)$ over the 3-volume of the star.

The main point to note here is that, given a particular choice for the EoS, the mass of the star $M$ and its radius $R$ are both one parameter families of the initial condition, i.e. central density $\epsilon_c$ (or the central pressure $p_c$).
Thus, the mass and radius can be treated as functions of the central density $M=M(\epsilon_c),R=R(\epsilon_c)$.
We can re-parameterize this family of equilibrium configurations in terms of $R$
(around points where the relation $R(\epsilon_c)$ is invertible), resulting in a single $M(R)$ relation, which is commonly referred to as the mass-radius curve. 

\subsection{Radial perturbations}

Once we have a background solution that represents a neutron star in equilibrium, we can study perturbations on this background. Following the general approach of \cite{Chandrasekhar1964}, we consider the Eulerian perturbation of the metric variables ($\Phi,\Lambda$) and fluid variables ($p,\epsilon,n$) around the background solution.
For example, we set $p(t,r)=p^{(0)}(r)+\delta p(t,r)$, where $p^{(0)}(r)$ is the background solution and $\delta p(t,r)$ is the Eulerian perturbation.
We also introduce Lagrangian perturbations, which describe how the fluid \textit{element} is affected by the perturbation.
We use the symbol $\Delta$ to refer to Lagrangian perturbations to contrast them with the Eulerian perturbations.
The Eulerian and Lagrangian perturbations are related by the Lagrangian displacement vector $\boldsymbol{\xi}^\mu$.
From spherical symmetry, we conclude that only the radial component of this displacement vector is nonzero, and thus, we label this nonzero component by just $\xi$.
Then, the Lagrangian and Eulerian perturbation for the scalar fluid quantities are related by
\begin{equation}\label{eq:lang}
    \Delta p=\delta p+\xi\frac{d}{dr}p^{(0)} \ .
\end{equation}

We next consider the perturbed Lagrangian displacement of the four-velocity.
The radial and time components of the four velocity of the fluid are related to first order in $\xi$ by $u^r=u^t\partial_t\xi\equiv u^t\dot\xi$. 
The constraint $u_\mu u^\mu=-1$ implies that the components of the velocity vector field up to first order is
\begin{equation}\label{eq:u}
    u^t=e^{\Phi^{(0)}}(1+\delta\Phi)\ ,\quad u^r=e^{\Phi^{(0)}}\dot\xi \ ,
\end{equation}
where the over head dot stands for time derivative $\dot{\xi} \equiv \partial_t \xi$.

We now present the linearly perturbed equations of motion. We first perturb the Euler equation [Eq.~\eqref{eq:euler}]. The radial component of the perturbed Euler equation is a second-order differential equation for the Lagrangian displacement,
\begin{equation}\label{eq:eul}
\begin{split}
    e^{2(\Lambda^{(0)}+\Phi^{(0)})}(\epsilon^{(0)}+p^{(0)})\ddot\xi=&-\partial_r\delta p+(\epsilon^{(0)}+p^{(0)})\partial_r\delta\Phi\\
    &+(\delta\epsilon+\delta p)\frac{d}{dr}\Phi^{(0)} \ .
\end{split}
\end{equation}
We next perturb the Einstein equations to find that the $(t,t)$, $(r,r)$, and $(t,r)$ components are
\begin{align}
    \label{eq:ddlam}
    \partial_r(re^{-2\Lambda^{(0)}}\delta\Lambda)
    =&    4\pi r^2\delta\epsilon \ ,    \\
    \label{eq:dphi}
    \partial_r\delta\Phi
    =&    -4\pi re^{2\Lambda^{(0)}}\delta p+\delta\Lambda\left(2\frac{d}{dr}\Phi^{(0)}-\frac1r\right)     \ ,\\
    \label{eq:ee-pert-tr}
    \frac2r\partial_t\delta\Lambda
    =&    -8\pi(\epsilon^{(0)}+p^{(0)})e^{2\Lambda^{(0)}}\dot\xi \ . 
\end{align}
Since only the perturbation functions $(\delta\Lambda,\xi)$ are time-dependent, it is straightforward to integrate Eq.~\eqref{eq:ee-pert-tr} over time.
Imposing the initial conditions $\delta\Lambda(0,r)=0$ and $\xi(0,r)=0$, we obtain
\begin{align}\label{eq:dlam}
    \delta\Lambda &=
    -4\pi(\epsilon^{(0)}+p^{(0)})re^{2\Lambda^{(0)}}\xi =\xi\frac{d}{dr}(\Phi^{(0)}-\Lambda^{(0)}) \ ,
\end{align}
where the second equality follows from Eq.~\eqref{eq:diff}.

We now perturb the fluid variables.
We start by linearly perturbing Eq.~\eqref{eq:consv},
which one can check results in a total time derivative.
We can integrate this time derivative by assuming the initial condition $\delta n(0,r)=0$, and after grouping terms, we obtain
\begin{equation}\label{eq:dn}
\begin{split}
    &\delta n+\frac{e^{-\Phi^{(0)}}}{r^2}\partial_r\left(n^{(0)}\xi r^2 e^{\Phi^{(0)}}\right)\\
    &\quad\quad\quad\quad +n^{(0)} \left[\delta\Lambda+\xi\frac{d}{dr}\left(\Lambda^{(0)}-\Phi^{(0)}\right)\right]=0 \ .  
\end{split}
\end{equation}
By Eq.~\eqref{eq:dlam}, the term multiplying $n^{(0)}$ is zero, so Eq.~\eqref{eq:dn} reduces to
\begin{equation}\label{eq:pertn}
    \delta n=-\frac{e^{-\Phi^{(0)}}}{r^2}\partial_r\left(n^{(0)}\xi r^2e^{\Phi^{(0)}}\right) \ .
\end{equation}
To obtain a useful expression for $\delta p$, we use Eq.~\eqref{eq:gamma}, which implies that 
\begin{equation}\label{eq:ppnn}
    \frac{\Delta p}{p^{(0)}}=\gamma\frac{\Delta n}{n^{(0)}} \ .
\end{equation}
Expanding the Lagrangian displacement using Eq.~\eqref{eq:lang}, we see that
\begin{equation}\label{eq:dp}
    \delta p=-\xi\frac{d}{dr}p^{(0)}-\gamma \, p^{(0)}\frac{e^{-\Phi^{(0)}}}{r^2}\partial_r\left(r^2\xi e^{\Phi^{(0)}}\right) \ .
\end{equation}
We see then that the perturbations to $n$ and $p$ are entirely prescribed by the background solution and the Lagrangian displacement $\xi$.

An expression for the perturbation $\delta\epsilon$ can be obtained in various ways. 
In \cite{Chandrasekhar1964}, $\delta\epsilon$ is found from Eq.~\eqref{eq:ddlam}.
To write $\delta\epsilon$ solely in terms of other fluid variables (which is more useful for the two-fluid case), we use the first law of thermodynamics [Eq.~\eqref{eq:first}], which when perturbed becomes
\begin{equation}
    \frac{\Delta\epsilon}{\Delta n}=\frac{\epsilon^{(0)}+p^{(0)}}{n^{(0)}} \ .
\end{equation}
Substituting in Eq.~\eqref{eq:dn} above and solving for $\delta\epsilon$ results in an expression in terms of the sound speed and its relation to $\gamma$ from Eq.~\eqref{eq:gcs}, namely
\begin{equation}\label{eq:de}
    \delta\epsilon=\delta p\frac{\epsilon^{(0)}+p^{(0)}}{\gamma \, p^{(0)}}=\delta p\frac{1}{c_s^2}=\delta p\frac{1}{\frac{\partial p}{\partial\epsilon}\eval_{s}} \ .
\end{equation}
Notice that the perturbation to the energy density diverges for a fixed perturbation in the pressure when the speed of sound vanishes (e.g.~during a first-order phase transition). 

\subsection{Sturm-Liouville Problem and Variational Principle}\label{ssec:var1}

In the previous subsection, we derived all the equations that are necessary to understand radial perturbations of a single-fluid star, but the equations are coupled and difficult to study in their current form. In order to understand the solution to these equations, one usually decouples them into a Sturm-Liouville problem to find the radial frequencies of perturbation, which is what we do in this subsection.  From now on, we drop the $(0)$ superscript on background quantities
because, in the following expressions, the only perturbation variable will be the Lagrangian displacement $\xi$.
Moreover, we use $\Phi'\equiv {d}\Phi/{dr}$ to simplify notation.

From substituting Eqs.~\eqref{eq:dphi},~\eqref{eq:dlam},~\eqref{eq:dp} and~\eqref{eq:de} into Eq.~\eqref{eq:eul} we can obtain a single differential equation for the Lagrangian displacement
\allowdisplaybreaks[4]
\begin{widetext}
\begin{equation}
\label{eq:diff-eqn-displacement-sturm-lioville-inter}
\begin{split}
    e^{2\Lambda+2\Phi}(\epsilon+p)\ddot\xi=&-\partial_r\left[-\xi(\epsilon+p)\Phi'-\gamma \, p\frac{e^{-\Phi}}{r^2}\partial_r\left(r^2\xi e^{\Phi}\right)\right]\\
    &-4\pi re^{2\Lambda}(\epsilon+p)\left[-\xi(\epsilon+p)\Phi'-\gamma \, p\frac{e^{-\Phi}}{r^2}\partial_r\left(r^2\xi e^{\Phi}\right)\right]\\
    &-4\pi re^{2\Lambda}(\epsilon+p)^2\xi\left(2\Phi'-\frac1r\right)\\
    &+\Phi'\left(\frac{\epsilon+p}{\gamma \, p}+1\right)\left[-\xi(\epsilon+p)\Phi'-\gamma \, p\frac{e^{-\Phi}}{r^2}\partial_r\left(r^2\xi e^{\Phi}\right)\right] \ .
\end{split}
\end{equation}
\end{widetext}
We can write Eq.~\eqref{eq:diff-eqn-displacement-sturm-lioville-inter} as a Sturm-Liouville equation by making the following redefinitions.
First, we redefine $\phi(t,r) \equiv r^2e^{\Phi(r)}\xi(t,r)$. Second, we multiply Eq.~\eqref{eq:diff-eqn-displacement-sturm-lioville-inter} by $e^{\Lambda-2\Phi}$ to find\footnote{We write the equation in the form of (26.19-21) of \cite{MisnerThorneWheeler1971}, where they are presented explicitly in Sturm-Liouville form.}
\begin{equation}\label{eq:1f}
    A(r)\ddot\phi=\partial_r \left[P(r)\partial_r\phi\right]-Q(r)\phi \ ,
\end{equation}
where
\begin{subequations}\label{eq:funcs1}
    \begin{align}
        A(r)&
        \equiv
        e^{3\Lambda-\Phi}\frac{(\epsilon+p)}{r^2}\label{eq:bA} \ ,
        \\
        P(r)
        &\equiv
        e^{\Lambda-3\Phi}\frac{\gamma \, p}{r^2}\ , \label{eq:bP} \\
        Q(r)
        &\equiv
        -e^{\Lambda-3\Phi}\frac{(\epsilon+p)}{r^2}\left[\Phi''-\frac2r\Phi'-(\Phi')^2\right] \nonumber        \\
        &\quad 
        +4\pi e^{3\Lambda-3\Phi}\frac{(\epsilon+p)^2}{r}\left(\Phi'-\frac1r\right) . \label{eq:Qb}
    \end{align}
\end{subequations}

We use Eq.~\eqref{eq:1f} as our master equation of motion for the radial fluid perturbations.
When we decompose $\phi$ into modes, $\phi(t,r)=e^{i\omega t}\zeta(r)$, 
we see that $\zeta(r)$ satisfies the following eigenvalue equation
\begin{equation}\label{eq:sl}
    \omega^2A(r)\zeta(r)=-\frac{d}{dr}\left(P(r)\frac{d}{dr}\zeta(r)\right)+Q(r)\zeta(r)\ .
\end{equation}
Importantly, both $A(r),P(r)$ are positive definite.
Imposing regularity of the solution at the origin, we find that
\begin{align}\label{eq:bc1}
    \lim_{r\to0}\frac{\zeta}{r^3} = {\rm{const.}} \ 
\end{align}
At the surface of the star, the Lagrangian displacement of the pressure must vanish. From Eq.~\eqref{eq:dp}, we conclude that
\begin{equation}\label{eq:boundaryc}
    \Delta p(R)=-\gamma \, p\frac{e^{-\Phi}}{r^2}\frac{d\zeta}{dr}\eval_R=0 \ 
\end{equation}
at the surface of the star.
Since $p(R)=0$, Eq.~\eqref{eq:boundaryc} essentially imposes that $d\zeta/dr$ (and thus $\zeta$) must be bounded at the stellar surface.

Equation \eqref{eq:sl} together with the boundary conditions of Eqs.~\eqref{eq:bc1} and \eqref{eq:boundaryc} describe a \textit{singular} Sturm-Liouville problem, because $P(r)$ vanishes at the surface.
We then conclude that the eigenvalues are real, and the corresponding eigenvectors form a complete orthogonal basis,
so that the general solution to the equation of motion in Eq.~\eqref{eq:1f} can be written as a linear sum of the eigenvectors of Eq.~\eqref{eq:sl}.
As long as all modes are real ($\omega_j^2>0$), the radial perturbations remain bounded, and we conclude the star is linearly stable.
We emphasize that \emph{completeness} is the property that ensures that such mode expansion exists, and thus it is a necessary condition to be confident on linear stability analysis \cite{DysonSchutz1979}. 

\subsection{Lower bound for eigenfrequencies}\label{ssec:bound}

As Eq.~\eqref{eq:sl} defines a singular Sturm-Liouville problem, we conclude that the eigenvalues (squared frequencies) $\omega^2$ are all real and bounded from below.
To set the stage for our two-fluid calculation, here we derive a lower bound to $\omega^2$ in terms of the other fluid quantities.

We begin by noting that, since Eq.~\eqref{eq:sl} together with the boundary conditions in Eqs.~\eqref{eq:bc1} and \eqref{eq:boundaryc} define a singular Sturm-Liouville problem,
the modes $\zeta$ are orthogonal on a Hilbert space whose inner product is defined by
\begin{equation}
    \label{eq:sturm-liouville-inner-product-space}
    \braket{\eta}{\zeta}
    \equiv
    \int_0^Rdr\, \bar\eta(r)A(r)\zeta(r) \ .
\end{equation}
We can then use this property to 
define a Rayleigh quotient, which can be used to derive a variational bound on the eigenvalue solutions to the problem.
For Eq.~\eqref{eq:sl}, the Rayleigh quotient is \cite{Chandrasekhar1964,HTWW1965}

\begin{equation}\label{eq:Rayleigh}
    I[\zeta]=\frac{\int_0^R dr\, P(r)\left|\frac{d}{dr}\zeta(r)\right|^2+Q(r)|\zeta(r)|^2}{\int_0^R dr\, A(r)|\zeta(r)|^2} \ .
\end{equation}
The extremal values of $I$ result in the eigenvalues $\omega^2$ of the differential equation Eq.~\eqref{eq:sl}. 
Although this result is well-known in Sturm-Liouville theory, and it has been used in previous studies of neutron star stability, we refer the unfamiliar reader to Sec.~11.1 of \cite{Strauss2009}, where such computation is carried out for the Laplace operator in detail\footnote{In~\cite{Strauss2009}, this approach is referred to as a minimum principle, but we will refer to it more generally as an extremal principle, since we have not yet proved that a minimum exists.}. We will not repeat this derivation here, since we will generalize it to multiple fluids in Sec.~\ref{ssec:variation} in detail. 

We next write the Sturm-Liouville problem [Eq.~\eqref{eq:sl}] in a more abstract form, which will prove useful when we examine the radial stability of multifluid stars.
First, we write the equation of motion Eq.~\eqref{eq:1f} in terms of two operators $\hat A$ and $\hat C$, namely
\begin{equation}\label{eq:wop}
    \hat A \, \ddot\phi+\hat C \, \phi=0 ,
\end{equation}
where the operators $A$ and $C$ are given in terms of the functions $A, P$ and $Q$ [see  Eq.~\eqref{eq:funcs1}] as
\begin{subequations}
\begin{align}
    \hat A \, \phi&\equiv A(r)\phi \ ,\\
    \hat C \, \phi&=-\partial_r\left(P(r)\partial_r\phi\right)+Q(r)\phi \ .
\end{align}
\end{subequations}
We then introduce the Lagrangian 
\begin{equation}
    L[\phi,\dot\phi]=\frac12\left[\int_0^R dr\, \dot{\bar\phi}\, \hat A \, \dot\phi -\bar\phi\, \hat C \, \phi\right] \ .
\end{equation}
Variation of ${L}$ with respect to $\zeta$ implies the equation of motion \eqref{eq:1f}.

The Lagrangian contains a term that is quadratic in the time derivative of the displacement,
minus another term that is quadratic on the displacement itself, which suggests we understand them as kinetic and potential energy terms. 
The kinetic and potential energies for the Lagrangian are then
\begin{align}
    T[\dot\phi]&=\frac12 \int_0^R dr\, A(r)\big|\dot\phi\big|^2 \label{eq:kinetic} \ ,\\
    V[\phi]&=\frac12\int_0^R dr\, \bar\phi \,\hat C \, \phi=\int_0^Rdr\, P(r)|\partial_r\phi|^2+Q(r)|\phi|^2 \,,
\end{align}
which agree with those presented in Appendix B of \cite{HTWW1965}.

Instead of working with the Rayleigh quotient to find the eigenvalues, we can work directly with the kinetic and potential energies. 
In the Hilbert space with inner product given by Eq.~\eqref{eq:sturm-liouville-inner-product-space}, the kinetic and potential energies can be written as
\begin{align}\label{eq:kinpot}
    T[\dot{\phi}] &= \frac{1}{2} \left<\dot{\phi}|\dot{\phi}\right>\,,
    \qquad
    V[\phi] =\frac12\ev{\mathcal V}{\phi} \ ,
\end{align}
where we have defined $\mathcal V\equiv \hat A^{-1}\hat C$; note that the factor of $\hat A^{-1}$ is necessary to cancel the factor of $A$ in the inner product.
The eigenfrequencies of the Sturm-Liouville problem defined by Eqs.~\eqref{eq:sl}, \eqref{eq:bc1} and \eqref{eq:boundaryc}
are, in fact, the eigenvalues of the potential energy operator $\mathcal V$.
This can be seen in that Eq.~\eqref{eq:sl} is then $\hat A(\mathcal V-\omega^2\mathbbm1)\zeta=0$, where $\mathbbm1$ is the identity operator. 

The Rayleigh quotient can also be written in terms of inner products. As shown in Appendix B of \cite{HTWW1965}, by taking the quotient of the potential energy by the kinetic energy [Eq.~\eqref{eq:kinetic}] evaluated at $\zeta$ instead of $\dot\phi$,
one finds that the Rayleigh quotient becomes
\begin{equation}\label{eq:action}
    I[\zeta]=\frac{V[\zeta]}{T[\zeta]}=\frac{\ev{\mathcal V}{\zeta}}{\braket\zeta} \ .
\end{equation}
When expressed in this way, the fact that the eigenvalues $\omega^2$ are given by the extrema of $I$ is called the Chandrasekhar variational principle.
Since the eigenvalues $\omega^2$ are (extremal) values of $I$, it is possible to analyze the properties of the eigenvalues of the Sturm-Liouville problem from the properties of $I$.
In particular, the stability condition $\omega_j^2>0$ for all $j$ can only be achieved if there is a minimum eigenvalue $\omega_0^2$ and if such a minimum eigenvalue is greater than $0$.
Such an eigenvalue is called the fundamental mode.

The existence of the fundamental mode can be proved by showing that $I$ as given in Eq.~\eqref{eq:Rayleigh} is bounded from below.
This is a standard result in Sturm-Liouville theory in a bounded domain \cite{RieszNagy1956} 
As we will carry out an analogous calculation when analyzing two-fluid (and more generally multi-fluid) solutions, we next derive an explicit lower bound for the mode frequencies.
We note that to our knowledge this is the first such derivation for relativistic single-fluid stars, although similar calculations have been performed for modes of rotating Newtonian stars in \cite{DysonSchutz1979,Hunter1977}.
Nevertheless, the explicit form of the bound we calculate ahead appears to be novel.

Let us begin by simplifying the calculation slightly. From writing the quotient as in Eq.~\eqref{eq:action}, it follows that a scaling of the vector $\zeta$ by a constant factor does not change the value of $I[\zeta]$.
Thus, we can assume, without loss of generality, that
\begin{align}
\label{eq:one-fluid-normalization-assumption}
    \braket \zeta=\int_0^Rdr\, A(r)|\zeta|^2=1.
\end{align}
Therefore, we need only check whether the numerator of Eq.~\eqref{eq:Rayleigh} is bounded from below.
First, the function $P(r)$ is positive-definite, and thus the first term
\begin{equation}
    \int_0^Rdr\, P(r)\left|\frac{d}{dr}\zeta\right|^2\geq0 \ ,
\end{equation}
and thus we can bound $I$ solely by
\begin{equation}
    I[\zeta]\geq \int_0^Rdr\, Q(r)|\zeta|^2 \ .
\end{equation}

Looking at the expression for $Q(r)$ in Eq.~\eqref{eq:Qb}, we have some terms that are positive-definite and some that are negative-definite.
The positive-definite term is the $-(\Phi')^2$ in the first line of Eq.~\eqref{eq:Qb}, while the entire second line is negative-definite, since $\Phi'\leq 0$ from Eq.~\eqref{eq:phi2}.
Just as we did above, we can bound the positive-definite term by zero, and we can write the remaining terms in the first line of Eq.~\eqref{eq:Qb} as a total derivative. Thus, we obtain the bound
\begin{equation}\label{eq:onefluidbound}
\begin{split}
    I[\zeta]\geq&-\int_0^Rdr\, e^{\Lambda-3\Phi}\frac{\epsilon+p}{r^2}|\zeta|^2 \left[r^2\frac{d}{dr}\left(\frac{\Phi'}{r^2}\right)\right]\\
    &-\int_0^Rdr\, e^{\Lambda-3\Phi}\frac{\epsilon+p}{r^2}|\zeta|^2
    \\
    &\qquad \qquad \times\left[4\pi re^{2\Lambda}(\epsilon+p)\left(\frac1r-\Phi'\right)\right] \ .
\end{split}
\end{equation}
Notice that the term in the second line is negative definite.

In order to make further progress, we must now massage the above expression further.
We can factor the second term in Eq.~\eqref{eq:onefluidbound} into a term that is the integrand of the inner product times some positive definite term.
More specifically, the supremum of the factor in square brackets in the second term of Eq.~\eqref{eq:onefluidbound} (which one can effectively think of as the maximum value) on the interval $(0,R)$ is a \textit{constant} that is always greater than (or equal to) the term in square bracket itself in the integration domain.
Thus, due to the negative sign, the second term in Eq.~\eqref{eq:onefluidbound} is greater than (or equal to) the same expression but with the supremum of the square bracket factored out of the integral, namely
\begin{align}\label{eq:onebound}
    &-\int_0^Rdr\, e^{\Lambda-3\Phi}\frac{\epsilon+p}{r^2}|\zeta|^2\times\left[4\pi re^{2\Lambda}(\epsilon+p)\left(\frac1r-\Phi'\right)\right]
    \nonumber \\
    \geq & 
    -\sup_{r\in(0,R)}\left[4\pi re^{2\Lambda}(\epsilon+p)\left(\frac1r-\Phi'\right)\right]
    \nonumber  \\
    &\qquad \times \int_0^Rdr\, e^{\Lambda-3\Phi}\frac{\epsilon+p}{r^2}|\zeta|^2 \ .
\end{align}
This inequality holds if the supremum in the second line does indeed exists.
The third line in the above expression is equal to $\int_0^Rdr\, A(r)|\zeta|^2=\braket\zeta=1$ via Eq.~\eqref{eq:one-fluid-normalization-assumption}.
Note that $\epsilon+p$ reaches it maximum value at the center of the star, since both $\epsilon$ and $p$ are monotonically decreasing as a consequence of the TOV equations.
We denote this maximal value by $\rho^{\rm max}\equiv \epsilon_c+p_c$, since $\epsilon+p$ goes to the baryon density in the Newtonian limit.
Hence, we are left with checking whether the supremum of $e^{2\Lambda}(1-r\Phi')$ exists.
The term $e^{2\Lambda}=({1-2m/r})^{-1}$ can only go to infinity if $m(r)=r/2$, but this violates the Buchdahl limit for a star, and thus, it does not occur.
Then, looking at Eq.~\eqref{eq:phi2}, we have that $\lim_{r\to0}\Phi'\sim r$ since $m\sim r^3$ when $r$ is near zero.
Hence, the supremum of the quantity in Eq.~\eqref{eq:onebound} exists and the inequality expressed holds true.

We can carry out a similar bound for the term in the first line of Eq.~\eqref{eq:onefluidbound}, which results in
\begin{equation}
\begin{split}
    &-\int_0^Rdr\, e^{\Lambda-3\Phi}\frac{\epsilon+p}{r^2}|\zeta|^2\times\left[r^2\frac{d}{dr}\left(\frac{\Phi'}{r^2}\right)\right]\\
    \geq &-\sup_{r\in(0,R)}\left[r^2\frac{d}{dr}\left(\frac{\Phi'}{r^2}\right)\right]\int_0^Rdr\, e^{\Lambda-3\Phi}\frac{\epsilon+p}{r^2}|\zeta|^2 \ .
\end{split}
\end{equation}
As $\Phi'\sim r$ near the origin, we have that $r^2\frac{d}{dr}(\Phi'/r^2)\sim -r^2(1/r^2)\sim 1$, which implies the supremum is bounded from above.

Putting the two terms from Eqs.~\eqref{eq:onefluidbound} and~\eqref{eq:onebound} together, we have that the bound for $I[\zeta]$ (recalling that the remaining integral is $1$ by Eq.~\eqref{eq:one-fluid-normalization-assumption}) is given by
\begin{equation}\label{eq:bound}
    I[\zeta]\geq -\!\!\!\!\sup_{r\in(0,R)}\left[r^2\frac{d}{dr}\left(\frac{\Phi'}{r^2}\right)\right]-4\pi \rho^{\rm max}\!\!\!\!\sup_{r\in(0,R)}\!\!\left[e^{2\Lambda}(1-r\Phi')\right].
\end{equation}
When one compares this bound to that found for Newtonian rotating stars,
one finds that the Newtonian limit of Eq.~\eqref{eq:bound} matches the non-rotating limit of Eq.~$(3.16)$ of \cite{DysonSchutz1979} (see also~\cite{Hunter1977}).
With this result, we have determined that a fundamental mode exists and in fact $\omega_0^2=\min_{\zeta\in\mathcal D(\mathcal V)}(I[\zeta])$, where $D(\mathcal V)$ is the domain of $\mathcal V$, defined below Eq.~\eqref{eq:kinpot}.
Then, the stability criteria for an equilibrium configuration is that $\omega_0^2=I[\zeta_0]>0$ for the eigenvector of the fundamental mode $\zeta_0$.
Equivalently, this can be written as $\ev{\mathcal V}{\zeta_0}=2V[\zeta_0]>0$.

\subsection{Stability criterion from Mass-Radius relation}

For single fluid stars, there is a one-to-one relation between solutions that satisfy $dM/dR=0$ (which we call ``extremum configurations''), and solutions that have at least one mode that is identically zero \cite{BardeenThorneMeltzer1966,HTWW1965}.
This relation is widely used to study the radial stability of stellar configurations, as computing $M(R)$ requires only solving the background equations of motion, instead of the more complicated radial perturbation equations.
As we discuss later though, while extremum configurations have at least one mode which is zero, there is no guarantee as to \emph{which} mode is zero, i.e.~the zero mode need not be the fundamental model. It is the zero crossings of the zero mode which delineate radially stable versus unstable stellar configurations.

First, we note that we can parameterize stellar solutions as $(M(\epsilon_c),R(\epsilon_c))$, because both the stellar mass $M$ and radius $R$ (for a fixed, single-parameter equation of state) only depend on the central density of the star.
Since $R(\epsilon_c)$ is a one-parameter family of functions, we can invert the relation (in open neighborhoods about points with ${dR}/{d\epsilon_c}\ne0$) to obtain another one-parameter family of equilibrium configurations $M(R)$.
Extrema of the two one-parameter families of the  mass (i.e. $M(\epsilon_c)$ and $M(R)$) are related to each other by the chain rule according to
\begin{equation}
    \frac{dM}{d\epsilon_c}\eval_{\epsilon^*}=\frac{dM}{dR}\eval_{R(\epsilon^*)}\frac{dR}{d\epsilon_c}\eval_{\epsilon^*} \ ,
\end{equation}
where $\epsilon^*$ is the central density at which the extremum occurs.
Hence, an extremum of $M(\epsilon_c)$ will always be an extremum of either $M(R)$ or $R(\epsilon_c)$, and moreover, extrema of $M(R)$ will also be extrema of $M(\epsilon_c)$
(as long as both $\frac{dR}{d\epsilon_c}$ and $\frac{d\epsilon_c}{dR}$ are nonzero at this point). Given all of this, one can use the extrema of $M(R)$ to study stability, instead of considering the extrema of $M(\epsilon_c)$.

A star is radially stable if and only if $\omega_0^2>0$. 
The argument we presented above shows that when $dM/dR=0$, \emph{some} oscillation mode has zero frequency, but this does not show that the fundamental mode is the one with zero frequency.
To address this issue (see e.g.~\cite{Oppenheimer1939}) one assumes that the equation of state in the low density regime must coincide with that of white dwarfs.
Given that white dwarfs are stable (as confirmed by astronomical observations), one then concludes that white dwarf stars sit on a stable branch, and thus all eigenfrequencies for them must be positive.
Hence, the first extrema (in fact maximum) of the $M(\epsilon_c)$ curve must be where the fundamental mode changes from a positive to a negative frequency.
We note that at the zeros of $dM/d\epsilon_c$, $dR/d\epsilon_c<0$ indicates a change of stability of an even mode, while $dR/d\epsilon_c>0$ is a change of stability for an odd mode \cite{BardeenThorneMeltzer1966}.
Using this fact, it is possible to determine which mode changes stability and whether they turn positive or negative.
For a more thorough discussion of the general relation between the zeros of $dM/d\epsilon_c$ and the sign of the radial modes, see \cite{BardeenThorneMeltzer1966} and chapter $6.8$ of \cite{ShapiroTeukolsky1983}.


We end by noting some implicit assumptions in the approach we outlined above.
The first assumption is that $\omega^2_0[\epsilon^*]$ is not itself at a local minimum (or maximum), such that $\omega^2_0$ will be positive (or negative) at both sides of the $\omega^2_0=0$ point.
Secondly, the one parameter family $M(\epsilon_c)$ must be continuously differentiable, which is guaranteed if the TOV equations themselves (the system Eqs.~\eqref{eq:hydro}, \eqref{eq:m} and \eqref{eq:phi2} together with the EoS $p(\epsilon)$) are continuously differentiable on the initial conditions. Note that this is not always true when considering equations of state with certain models for first-order phase transitions. 
The third assumption is that ${dM}/{dR}=0$ divides the configuration space into two disjoint regions of parameter space: a ``stable'' region and an ``unstable'' region.
This is a weak assumption for single-fluid stars, but for two-fluid (or more fluids) stars it is less obvious that this always holds; nevertheless, we will continue to assume that it holds true in the two-fluid case as well, when properly generalized to a higher-dimensional initial condition space.

\section{Lagrangian Perturbations for Two Fluids}\label{sec:lang}

In this section, we derive the background and radial perturbative equations of motion for a two-fluid relativistic star.
Our main result is that we can write the radial perturbative equations as a system of symmetric second order equations that are analogous to those for a single-fluid star.
Our approach is similar to what of \cite{Kain2020}, except that this reference did not rewrite the equations in a symmetric form, which is a crucial step in our later calculations.
Reference \cite{Comer1999} additionally considered the radial perturbations of two-fluid stars, with the inclusion of fluid interactions terms. 
The approach in \cite{Comer1999} was later numerically studied in \cite{Leung2012}.
We do not take this latter approach, as \cite{Comer1999} did not write the equations of motion in a symmetric form either.


We consider two non-interacting, relativistic fluids, and write the stress-energy tensor as
\begin{align}
    \label{eq:multiple-fluid-perfect}
    T_{\mu\nu}
    =
    \sum_{X=1}^2 T_{\mu\nu}^{(X)}
    ,
\end{align}
where $T_{\mu\nu}^{(X)}$ refers to a single-fluid stress energy tensor [see Eq.~\eqref{eq:single-fluid-perfect}].
Each fluid has their own fluid variables $(\epsilon_X,p_X,n_X)$, together with a one parameter family EoS for each fluid $p_X(\epsilon_X)$.
We index the two fluids $X \in (1, 2)$, with the capital roman letters $X, Y, Z$, etc indicating the fluid.

Each fluid obeys the Euler equation, the first law of thermodynamics, and the number density conservation independently because we assume they only interact gravitationally.
Thus, we can generalize Eqs.~\eqref{eq:consv} and \eqref{eq:first} simply by setting $p\to p_X,\epsilon\to\epsilon_X,n\to n_X$.
We can similarly define the speed of sound and the adiabatic index for each fluid by subscripting the appropriate variable with the index '$X$'.
In particular, we have
\begin{gather}
    \frac{d\epsilon_X}{dn_X}=\frac{\epsilon_X+p_X}{n_X} \ , \label{eq:2first}\\
    c_{s,X}^2=\frac{d p_X}{d \epsilon_X} \ ,\\
    \gamma_X=\frac{d\log p_X}{d\log n_X} \ .
\end{gather}
As for a single fluid, we obtain the following relation between the speed of sound and the adiabatic index:
\begin{equation}\label{eq:relc}
    \gamma_Xp_X=(\epsilon_X+p_X)c_{s,X}^2 \ .
\end{equation}

\subsection{Background Equations}\label{ssec:back2}

The background spacetime will continue to be treated as static and spherically symmetric, so that its line element can be represented by Eq.~\eqref{eq:metric}.
As the background is static and the four-velocity of the fluids is timelike, $\left(u_X\right)_{\mu} \left(u_X\right)^{\mu}=-1$, we set $u_X^\mu=(e^{\Phi},0,0,0)$ for each fluid --that is, the background fluid vectors are identical to those of a static single-fluid star.
The independent conservation of each stress-energy tensor then leads to the two-fluid generalization of the Euler equation in Eq.~\eqref{eq:hydro}, namely
\begin{align}
    \frac{d}{dr}p_X=(\epsilon_X+p_X)\frac{d}{dr}\Phi \label{eq:px} \ ,
\end{align}
which holds for each fluid independently.

The background equations of motion for the metric variables $\Lambda$ and $\Phi$ remain essentially unchanged from their single-fluid forms (Eqs.~\eqref{eq:lam} and \eqref{eq:phi}, respectively), except that $\epsilon\to\sum_X\epsilon_X$ and $p\to\sum_Xp_X$. 
As for the case of a single fluid, we define $m(r)$ from $\Lambda(r)$ using Eq.~\eqref{eq:defm}.
The equations of motion for the mass aspect $m(r)$ and the metric function $\Phi(r)$ are then
\begin{align}
    \frac{d }{dr}m&=4\pi r^2\sum_X\epsilon_X\label{eq:2m} \ , \\
    \frac{d}{dr}\Phi&=-\frac{1}{r^2}\frac{1}{1-\frac{2m}{r}}\left(m+4\pi r^3\sum_X p_X\right)\label{eq:2phi2} \ .
\end{align}
The generalization of Eq.~\eqref{eq:diff} for a multi-fluid star is
\begin{align}\label{eq:2diff}
    \frac{d}{dr}\Lambda-\frac{d}{dr}\Phi=4\pi re^{2\Lambda}\sum_X(\epsilon_X+p_X) \ . 
\end{align}

We can then integrate the system of Eqs.~\eqref{eq:px}, \eqref{eq:2m} and \eqref{eq:2phi2} in almost the same way as for a single-fluid star. First, we need to specify a one-parameter EoS $p_X(\epsilon_X)$ for each fluid.
Then, we integrate Eq.~\eqref{eq:2m} together with the two pressure equations Eq.~\eqref{eq:px}, substituting for $d\Phi/dr$ using Eq.~\eqref{eq:2phi2} in each case. 
The initial conditions and the termination conditions, however, must be chosen carefully.  
The initial condition for the mass remains $m(0)=0$, while the total central pressure of the star is equal to the sum of the central pressures $p_X^c$ of each fluid, which in turn are functions of the central densities $\epsilon_X^c$ and equation of state $p_X^c = p_X\left(\epsilon_X^c\right)$.
Unlike the case for a single-fluid star, a two-fluid star generally has two ``outer'' boundaries, as the pressure can vanish for each fluid at different radii.
Without loss of generality, we assume that the pressure of fluid ``1'' vanishes at a radius smaller than or equal to the radius of fluid ``2''; that is $p_1$ vanishes at $r=R_1$, and $p_2$ vanishes at $R_2\geq R_1$.
We can then integrate the two-fluid TOV equations until radius $R_1$, and then evolve only the single-fluid TOV equations for fluid two until we reach radius $R_2$.
Once we have obtained $m$, $\epsilon_X$, and $p_X$, we can solve for $\Phi(r)$ by integrating from the outer boundary $R_2$ to the center of the star, using the boundary condition $\Phi(R_2)=-\frac12\log\left(1-\frac{2m(R_2)}{R_2}\right)$.
We impose continuity across the boundary $r=R_1$.
The number densities are obtained from integrating the first law of Eq.~\eqref{eq:2first} from the outer boundary to the center with boundary condition $n_X(R_X)=0$.
The total mass is obtained from $M=m(R_2)$, and the total particle numbers $N_X$ can be obtained from integrating $n_X$ over the $3$-volume of the star.
Finally, we note that the above equations and algorithmic procedure can be applied essentially unchanged for an arbitrary number of fluids (c.f. \cite{Kain2020});
we discuss the arbitrarily many fluids case in Sec.~\ref{sec:multi}.


\subsection{Radial Perturbations}\label{ssec:2pert}


Here we summarize the perturbative equations of motion for a multi-fluid star.
Our main results are Eqs.~\eqref{eq:2dlam}, \eqref{eq:2dphi}, \eqref{eq:2de}, \eqref{eq:2dp} and \eqref{eq:2eul},
which are the same as equations (43), (44), (46), (47) and (50) in \cite{Kain2020}. We rederive these equations here for consistency and to establish notation.

As for the single-fluid case, we denote the Eulerian perturbation of a field/fluid quantity via $\delta$; e.g. $\Phi(t,r)=\Phi^{(0)}(r)+\delta\Phi(t,r)$ for field variables and $p_X(t,r)=p_X^{(0)}(r)+\delta p_X(t,r)$ for fluid variables.
To define the Lagrangian perturbation for each fluid, we define separate displacement vectors $\xi_X$ for each fluid.
For example, the Lagrangian perturbation of the pressure for fluid $X$ is
\begin{equation}
    \Delta p_X=\delta p_X+\xi_X\frac{d}{dr}p_X^{(0)} \ .
\end{equation}
We also note that when we expand to first order, Eq. \eqref{eq:u} will have the same functional form as in the single-fluid case, but with $u\to u_X$ and $\xi\to \xi_X$.

The radial perturbative Euler equation [Eq. \eqref{eq:eul}] for a multi-fluid star ends up taking essentially the same form as it does for a single-fluid star.
The equation of motion for the fluid displacement vector $\xi_X$ is
\begin{align}\label{eq:2eul}
    e^{2(\Lambda^{(0)}+\Phi^{(0)})}(\epsilon_X^{(0)}+p_X^{(0)})\ddot\xi_X=&-\partial_r\delta p_X+(\epsilon_X^{(0)}+p_X^{(0)})\partial_r\delta\Phi \nonumber\\
    &+(\delta\epsilon_X+\delta p_X)\frac{d}{dr}\Phi^{(0)} \ .
\end{align}
The two-fluid generalization of Eq. \eqref{eq:dphi} simply involves setting $\delta p\to\sum_X\delta p_X$ as follows
\begin{align}\label{eq:2dphi}
    \partial_r\delta\Phi=-4\pi re^{2\Lambda^{(0)}}\left(\sum_X\delta p_X\right)+\delta\Lambda\left(2\frac{d}{dr}\Phi^{(0)}-\frac1r\right) \ .
\end{align}
Likewise, the generalization of Eq. \eqref{eq:dlam} is 
\begin{align}\label{eq:2dlam}
    \delta\Lambda&=-4\pi re^{2\Lambda^{(0)}}\sum_X(\epsilon_X^{(0)}+p_X^{(0)})\xi_X \nonumber\\
    &=
    \left(\frac{d}{dr}\Phi^{(0)}-\frac{d}{dr}\Lambda^{(0)}\right)
    \bar\xi \ ,
\end{align}
where $\bar\xi$ is a ``weighted average'' of the Lagrangian displacements defined as
\begin{align}\label{eq:barxi}
    \bar\xi\equiv \frac{\sum_X(\epsilon_X^{(0)}+p_X^{(0)})\xi_X}{\sum_X(\epsilon_X^{(0)}+p_X^{(0)})} \ .
\end{align}





We next turn to deriving the equations of motion for the perturbed fluid variables $p_X$ and $\epsilon_X$.
The form of Eq.~\eqref{eq:dn} is essentially the same for a multifluid star, except that for a multifluid star $\delta\Lambda+\xi_X(\frac{d}{dr}\Lambda^{(0)}-\frac{d}{dr}\Phi^{(0)})\ne 0$ as a result of Eq.~\eqref{eq:2dlam}.
The generalization of Eq.~\eqref{eq:dn} for a two-fluid star is then
\begin{align}
    \delta n^X
    &=
    -
    \frac{e^{-\Phi^{(0)}}}{r^2}\partial_r\left(n_X^{(0)}r^2\xi_X e^{\Phi^{(0)}}\right)
    \nonumber\\
    &
    +
    n_X^{(0)}\left(\frac{d}{dr}\Phi^{(0)}-\frac{d}{dr}\Lambda^{(0)}\right)(\xi_X-\bar\xi) 
    \ .
\end{align}
Written this way, it is straightforward to see how for one fluid the above equation reduces to the simpler expression in Eq.~\eqref{eq:pertn}, because then $\bar\xi=\xi_X$.
We can then use the definition of $\gamma_X$ to obtain the perturbations for $\delta p_X$ just as we did in the single-fluid case [see Eq.~\eqref{eq:ppnn}]. 
We obtain
\begin{equation}\label{eq:2dp}
\begin{split}
    \delta p^X=&-\xi_X\frac{d}{dr}p_X^{(0)}-\frac{\gamma_X p_X^{(0)}e^{-\Phi^{(0)}}}{r^2}\partial_r\left(r^2\xi_Xe^{\Phi^{(0)}}\right)\\
    &+\gamma_Xp_X^{(0)}\left(\frac{d}{dr}\Phi^{(0)}-\frac{d}{dr}\Lambda^{(0)}\right)(\xi_X-\bar\xi) \ .
\end{split}
\end{equation}
We can get a relation between $\delta p_X$ and $\delta\epsilon_X$ in the same way as in Eq.~\eqref{eq:de} for the one fluid case. 
We find
\begin{equation}\label{eq:2de}
    \delta\epsilon_X=\delta p_X\frac{\epsilon_X^{(0)}+p_X^{(0)}}{\gamma_Xp_X^{(0)}} \ .
\end{equation}

\subsection{Radial Lagrangian fluid displacement}

Up until now, the equations of motion we derived hold for a star made up of an arbitrary number of fluids.
Here we focus on a two-fluid star and show how the equations of motion for the radial Lagrangian displacement can be simplified and written in a symmetric form.
Having said that, the equations of motion we will present for a two fluid star can easily be generalized for a star made up of $N$ non-interacting fluids, as we briefly discuss in Sec.~\ref{sec:multi}.
%
To simplify our notation, we drop the $(0)$ superscript for background quantities,
and take $\Phi'\equiv \frac{d}{dr}\Phi$, just as we did for a single fluid.

Substituting Eq.~\eqref{eq:2dphi}, \eqref{eq:2dlam}, \eqref{eq:2dp} and \eqref{eq:2de} into the perturbed Euler equation Eq.~\eqref{eq:2eul} results in a second-order equation of motion for each of the $\xi_X$.
Defining the scaled Lagrangian displacements $\phi_X\equiv r^2e^{\Phi}\xi_X$ and multiplying the equations by $e^{\Lambda-2\Phi}$ as in Sec.~\ref{sec:onefluid} gives us
\begin{equation}\label{eq:op}
    \hat A_X\ddot\phi_X+\hat C_X\phi_X+\hat D_X\phi_Y=0 \ ,
\end{equation}
where $(\hat A_X,\hat C_X,\hat D_X)$ are operators given by
\begin{subequations}
    \begin{gather}
        \hat A_X \ddot\phi_X\equiv A_X(r)\ddot\phi_X \ ,\\
        \hat C_X\phi_X\equiv -\partial_r(P_X(r)\phi_X)+[Q_X(r)+R(r)]\phi_X \ ,\\
        \hat D_X\phi_Y\equiv S(r)\phi_Y+\partial_r(f_X(r)\phi_Y)-f_Y(r)\partial_r\phi_Y \ . \label{eq:dx}
    \end{gather}
\end{subequations}
The functions $[A_X(r),P_X(r),Q_X(r),S(r),R(r),f_X(r)]$ are derived in detail in Appendix \ref{app:symm} and they are explicitly given by
\begin{subequations}\label{eq:2funcs}
    \begin{align}
        A_X(r)=&e^{3\Lambda-\Phi}\frac{(\epsilon_X+p_X)}{r^2} \ ,\label{eq:2AX}\\
        P_X(r)=&e^{\Lambda-3\Phi}\frac{\gamma_Xp_X}{r^2} \ , \label{eq:2PX}\\
        Q_X(r)=&-e^{\Lambda-3\Phi}\frac{(\epsilon_X+p_X)}{r^2}\left[\Phi''-\frac2r\Phi'-(\Phi')^2\right]  \nonumber\\
        &+\frac{4\pi}{r} e^{3\Lambda-3\Phi}(\epsilon_X+p_X)^2\left(\Phi'-\frac1r\right)-\frac{d}{dr}f_X \ , \label{eq:2Qb}\\
        R(r)=&4\pi\frac{e^{3\Lambda-3\Phi}}{r}\Phi'(\epsilon_X+p_X)(\epsilon_Y+p_Y) \nonumber\\
        &+16\pi^2e^{5\Lambda-3\Phi}\gamma_Xp_X(\epsilon_Y+p_Y)^2 \nonumber\\
        &+16\pi^2e^{5\Lambda-3\Phi}\gamma_Yp_Y(\epsilon_X+p_X)^2 \ , \label{eq:Rb}\\
        S(r)=&-4\pi \frac{e^{3\Lambda-3\Phi}}{r^2}(\epsilon_X+p_X)(\epsilon_Y+p_Y) \nonumber\\
        &-16\pi^2e^{5\Lambda-3\Phi}\gamma_Xp_X(\epsilon_Y+p_Y)^2 \nonumber \\
        &-16\pi^2e^{5\Lambda-3\Phi}\gamma_Yp_Y(\epsilon_X+p_X)^2 \ , \label{eq:Sb}\\
        f_X(r)=&\frac{4\pi}{r}e^{3\Lambda-3\Phi}\gamma_Xp_X(\epsilon_Y+p_Y) . \label{eq:fx}
    \end{align}
\end{subequations}
Observe that these functions depend only on the background fields --i.e.~they are functions of just the radial coordinate $r$.
Compared to Eq.~\eqref{eq:funcs1}, $A_X(r)$ and $P_X(r)$ keep the same essential functional form, while $Q_X(r)$ is modified due to the presence of the multiple fluids, which interact with each other only through the gravitational fields.
The functions $[S(r),R(r),f_X(r)]$ all arise due to the gravitational interaction between the two fluids.
Observe also that $[A_X(r),P_X(r),f_X(r)]$ are all positive definite, and that $S(r)$ is negative definite.
Moreover, the functions $[f_X(r),S(r),R(r)]$ are all identically zero when $r>R_1$ (assuming $R_2\geq R_1$), while $Q_X(r)$ goes to the one-fluid expression in this region.
In Sec.~\ref{sec:energy}, we will use Eq.~\eqref{eq:op} to prove radial mode stability. 


We now rewrite the radial perturbative equations as an eigenvalue problem.
First, we assume that we can separate the solution in the form $\phi_X\left(t,r\right)=\zeta_X\left(r\right)\Theta_X\left(t\right)$.
Then, the radial perturbation equations demand that the two fluid elements have the same time dependence $\Theta_X\left(t\right)$ (up to an additive constant). 
We then decompose $\phi_X\left(t,r\right)= e^{i \omega t} \zeta_X\left(r\right)$, since the frequencies of oscillation of the two fluids must be the same.
Substituting this decomposition into Eq.~\eqref{eq:op}, we obtain 
\begin{equation}\label{eq:2sl}
    \omega^2\hat A_X\zeta_X=\hat C_X\zeta_X+\hat D_X\zeta_Y \ .
\end{equation}
Equation~\eqref{eq:2sl} represents a generalized eigenvalue problem for the mode frequencies $\omega^2$, given a set of boundary conditions for the variables $\zeta_X$.

The boundary condition for $\zeta_X$ is similar to that of a single fluid.
First, near the origin we must have $\xi_X(r\to0)\sim r$, which implies
\begin{equation}\label{eq:zero_boundary_2f}
    \lim_{r\to0}\frac{\zeta_X}{r^3}={\rm const}.
\end{equation}
We have two ``outer'' boundaries at the radii $R_1$ and $R_2$ and, as we mentioned before, we assume that $R_1<R_2$ without loss of generality. 
At each radial boundary we require the Lagrangian perturbation of the pressure for the fluid that vanishes at that boundary to go to zero
\begin{align}
    \Delta p_X\Big|_{r=R_X}
    =
    0
    .
\end{align}
At the inner radius $r=R_1$, we impose that all fluid quantities for fluid $2$ are continuous across that boundary.

The radial Lagrangian displacement for the pressure at the outer boundary $r=R_2$ takes the same form as it does for a single fluid
[Eq.~\eqref{eq:boundaryc}].
The radial Lagrangian displacement at the inner boundary though is modified as follows:
\begin{equation}\label{eq:outer_boundary_2f}
\begin{split}
    \Delta p_1(R_1)=\Bigg[&-\gamma_1p_1\frac{e^{-\Phi}}{r^2}\frac{d\zeta_1}{dr}\\
    &+4\pi \gamma_1p_1(\epsilon_2+p_2)\frac{e^{2\Lambda-\Phi}}{r}(\zeta_2-\zeta_1)\Bigg]\eval_{R_1}=0 \ .
\end{split}
\end{equation}
Since $p_1\to 0$ at $R_1$, this boundary condition, in turn, is taken to be $(\zeta_1,\frac{d\zeta_1}{dr},\zeta_2)$ remaining finite at $R_1$.
Moreover, we require continuity for $\zeta_2$ and $\frac{d\zeta_2}{dr}$ at $R_1$.

Even though Eqs.~\eqref{eq:2sl}, \eqref{eq:zero_boundary_2f} and \eqref{eq:outer_boundary_2f} define a generalized eigenvalue problem for the mode frequencies, it is \textit{not} a Sturm-Liouville problem, unlike in the single-fluid case of Sec.~\ref{ssec:var1}. This is because of the second term in the right-hand side of Eq.~\eqref{eq:2sl}, which couples the two fluids together. Therefore, to establish that the eigenvalues $\omega^2$ are real and bounded from below, one must prove so, as we do below. 

\section{Variational Principle for the radial modes}\label{sec:energy}

In this section, we extend the results of Secs.~\ref{ssec:var1} and \ref{ssec:bound} to two fluids, which is needed to establish radial-mode stability.
Our analysis is based on finding a canonical energy of the perturbations, analogous to what was done for two-fluid Newtonian stars in \cite{AnderssonComerGrosart2004}.
The canonical energy is then used to obtain an action principle for the equations of motion and a variational principle for the eigenmodes of the two-fluid star, equivalent to the Chandrasekhar variational principle of single fluid stars.

\subsection{Deriving the Canonical Energy}\label{ssec:canon}


In this subsection, we work in the time domain --that is, we will not Fourier decompose functions in time.
We introduce the standard $L^2$ inner product on $(0,R_2)$, namely
\begin{equation}
    \label{eq:unweighed-inner-product}
    \langle f,g\rangle=\int_0^{R_2}dr\, \bar f \; g \ ,
\end{equation}
where the overhead bar stands for complex conjugation.
In general, up to the outer radius, we use the same inner product for both fluid perturbations.

Following the derivation in \cite{AnderssonComerGrosart2004}, we start by defining a symplectic structure $W$ on the phase space of solutions.
We divide the symplectic product into a symplectic product for each fluid $W_X$
\begin{equation}\label{eq:symplec}
    W_X(\psi_X,\phi_X)=\langle \psi_X,\hat A_X\dot\phi_X\rangle-\langle \hat A_X\dot\psi_X,\phi_X\rangle \ .
\end{equation}
The full symplectic structure is then the sum of the symplectic structures for each fluid,
\begin{equation}
    W[\boldsymbol\psi,\boldsymbol\phi]=W_1(\psi_1,\phi_1)+W_2(\psi_2,\phi_2) \ ,
\end{equation}
where $\boldsymbol\phi=(\phi_1,\phi_2)$ and $\boldsymbol\psi = (\psi_1,\psi_2)$.

Given this symplectic structure, we can now state the following theorem. 
\begin{theorem}\label{th:symplectic}
    The symplectic structure $W[\boldsymbol\psi,\boldsymbol\phi]$ is a conserved quantity in time when evaluated on solutions of the equations of motion, Eq.~\eqref{eq:op}.
\end{theorem}

\begin{proof}

Let $\boldsymbol\phi=(\phi_1,\phi_2)$ and $\boldsymbol\psi=(\psi_1,\psi_2)$ be two solutions to the equations of motion.
Then, the time derivative of either of the symplectic structures results in 
\begin{equation}\label{eq:dwx}
\begin{split}
    \frac{d}{dt}W_X&=\langle \dot\psi_X,\hat A_X\dot\phi_X\rangle+\langle\psi_X,\hat A_X\ddot\phi_X\rangle
     \\
    &-\langle \hat A_X\ddot\psi_X,\phi_X\rangle-\langle \hat A_X\dot\psi_X,\dot\phi_X\rangle\\
    &=-\langle \psi_X,\hat C_X\phi_X+\hat D_X\phi_Y\rangle+\langle \hat C_X\psi_X+\hat D_X\psi_Y,\phi_X\rangle\\
    &=\langle \hat D_X\psi_Y,\phi_X\rangle-\langle \psi_X,\hat D_X\phi_Y\rangle \ .
\end{split}
\end{equation}
To go from the first equality to the second equality in the above equation, we used that $A_X$ is self adjoint, i.e.~$\langle \psi_X,A_X\phi_X\rangle = \langle A_X\psi_X,\phi_X\rangle$, and we substituted in Eq.~\eqref{eq:op}.
To go from the second to the equality the third line to the fourth, we used that $C_X$ is self adjoint and the inner product is distributive.
One point to highlight is that $C_1$ is self-adjoint on the interval $(0,R_2)$ because
\begin{equation}
    \int_0^{R_2} dr\, \bar\phi_1\hat C_1\phi_1=\int_0^{R_1}dr\, \bar\phi_1\hat C_1\phi_1 \ ,
\end{equation}
which follows from the fact that all functions $[P_1(r),Q_1(r),f_1(r),R(r)]$ [see Eq.~\eqref{eq:2funcs}] that enter the definition of $C_1$ vanish for all $r>R_1$.

In general, $\langle \hat D_X\psi_Y,\phi_X\rangle\ne\langle \psi_X,\hat D_X\phi_Y\rangle$ so that $\frac{d}{dt}W_X$ (as given in Eq.~\eqref{eq:dwx}) is nonzero.
However, the sum of the time derivatives of the two symplectic structures is
\begin{align}\label{eq:WW}
    \frac{dW}{dt}\left[\boldsymbol\psi,\boldsymbol\phi\right]
    =&\langle \hat D_1\psi_2,\phi_1\rangle
    -\langle \psi_2,\hat D_2\phi_1\rangle \nonumber\\
    +&\langle \hat D_2\psi_1,\phi_2\rangle-\langle\psi_1,\hat D_1\phi_2\rangle \ ,
\end{align}
and we will now show that this does vanish.
Let us write the first line explicitly using the expression for the operators $D_X$ in Eq.~\eqref{eq:dx},
\begin{equation}\label{eq:difd}
\begin{split}
        &\langle \hat D_1\psi_2,\phi_1\rangle-\langle\psi_2,\hat D_2\phi_1\rangle\\
        =&\int_0^{R_1} dr\, \left[S\bar\psi_2\phi_1+\partial_r(f_1\bar\psi_2)\phi_1-f_2\bar\psi_2\partial_r\phi_1\right]\\
        -&\int_0^{R_1} dr\, \left[\bar\psi_2S\phi_1+\bar\psi_2\partial_r(f_2\phi_1)-\bar\psi_2f_1\partial_r\phi_1\right] \ .
\end{split}
\end{equation}
The integral runs over $(0,R_1)$ because $[S(r),f_1(r),f_2(r)]$ vanish when $r>R_1$.
We note that the terms in Eq.~\eqref{eq:difd} which multiply $S$ cancel each other out.
We next integrate by parts using the boundary conditions, and find that
\begin{equation}
\label{eq:simplify-expr-fetze}
\begin{split}
    \int_0^{R_1}dr\, \partial_r(f_1\bar\psi_2)\phi_1=&f_1\bar\psi_2\phi_1\eval_{0}^{R_1}-\int_0^{R_1}dr\, \bar\psi_2f_1\partial_r\phi_1\\
    =&-\int_0^{R_1} dr\, \bar\psi_2f_1\partial_r\phi_1 \ ,
\end{split}
\end{equation}
where the boundary term vanishes because $[\phi_1(r),\psi_1(r)] \propto r^3$ near $r=0$, while at $r=R_1$ we have that $[\phi_1(r), \psi_1(r)]$ are regular, while $\lim_{r\to R_1}f_1=0$.
Thus, substituting into Eq.~\eqref{eq:difd}, we find that
\begin{equation}
\int_0^{R_1}dr\,\left[\partial_r\left(f_1\bar{\psi}_2\right)\phi_1 + \bar{\psi}_2f_1\partial_r\phi_1\right]=0 \ .
\end{equation}
This means that the second term in the first line of Eq.~\eqref{eq:difd} cancels exactly with the last term in the second line of this equation.
The same cancellation occurs among the remaining terms in Eq.~\eqref{eq:difd},
and therefore, Eq.~\eqref{eq:difd} vanishes exactly.
From the symmetry of index interchange, we conclude that
\begin{equation}\label{eq:sym}
    \langle \hat D_X\psi_Y,\phi_X\rangle=\langle\psi_Y,\hat D_Y\phi_X\rangle \ ,
\end{equation}
for any $X \in (1,2)$ and $Y \neq X$.
Therefore, both lines of Eq.~\eqref{eq:WW} vanish exactly, proving the statement of the theorem.

\end{proof}

Having defined a symplectic structure, we can now define an energy.
We note that if $\phi=(\phi_1,\phi_2)$ is a solution to the equations of motion [Eq.~\eqref{eq:op}], then $\dot\phi$ is also a solution.
Therefore, by Theorem \ref{th:symplectic}, the quantity $E_c[\boldsymbol\phi,\dot{\boldsymbol\phi}]\equiv -\frac12W[\boldsymbol\phi,\dot{\boldsymbol\phi}]$ is conserved in time. 
Following \cite{AnderssonComerGrosart2004,Prabhu_2016,Friedman1978}, we define the conserved energy functional in that way, which we can expand using the equations of motion to find 
\begin{align}\label{eq:2energy}
    E_c[\boldsymbol\phi,\dot{\boldsymbol\phi}]
    =&
    \frac12\left[\langle\dot\phi_1,\hat A_1\dot\phi_1\rangle+\langle\dot\phi_2,\hat A_2\dot\phi_2\rangle\right]
    \nonumber\\
    +&\frac12\left[\langle \phi_1,\hat C_1\phi_1\rangle+\langle\phi_2,\hat C_2\phi_2\rangle\right]
    \nonumber\\
    +&\frac{1}{2}\left[\langle\phi_1,\hat D_1\phi_2\rangle+\langle\phi_2,\hat D_2\phi_1\rangle\right]  \  ,
\end{align}
where we have used the fact that $A_X$ is self-adjoint.
We can think of the first line of Eq.~\eqref{eq:2energy} as representing the ``kinetic energy'' $T$ of the two-fluid system, while the second and third lines constitute the ``potential energy'' $V$ of the system.
That is,
\begin{align}
    T[\dot{\boldsymbol\phi}]
    &\equiv
    \frac12\left[\langle \dot\phi_1,\hat A_1\dot\phi_1\rangle+\langle\dot\phi_2,\hat A_2\dot\phi_2\rangle\right] 
    ,\\
    V[\boldsymbol\phi]
    &\equiv
    \frac12\big[
        \langle \phi_1,\hat C_1\phi_1\rangle+\langle \phi_2,\hat C_2\phi_2\rangle
        \nonumber\\
        &\qquad+\langle\phi_1,\hat D_1\phi_2\rangle+\langle\phi_2,\hat D_2\phi_1\rangle
    \big] \ .
\end{align}
We emphasize that the interaction between the two fluids only shows up in the potential energy.

\subsection{Variational Principle}\label{ssec:variation}

Having defined a kinetic and potential energy, we can now define a Lagrangian
\begin{equation}\label{eq:ii}
    L[\boldsymbol\phi,\dot{\boldsymbol\phi}]=T[\dot{\boldsymbol\phi}]-V[\boldsymbol\phi] \ ,
\end{equation}
and an action
\begin{align}
\label{eq:action_principle}
    S[\boldsymbol\phi,\dot{\boldsymbol\phi}]&=\int_{t_1}^{t_2} dt\,  L
    \nonumber
    \\
    &=\frac12\int_{t_1}^{t_2} dt\, \sum_X\Big(\langle \dot\phi_X,\hat A_X\dot \phi_X\rangle-\langle \phi_X,\hat C_X\phi_X\rangle
    \nonumber\\
    &\qquad\qquad\qquad-\langle \phi_X,\hat D_X\phi_Y\rangle\Big) \ .
\end{align}

Varying Eq.~\eqref{eq:action_principle} gives us the equations of motion Eq.~\eqref{eq:wop}.
To verify this, we follow the standard steps in the calculus of variations (e.g. Secs.~11.1 and 11.4 of \cite{Strauss2009}).
First, we set $\phi_X\to \phi_X+\varepsilon \; \psi_X$ for both fluids, where $\psi_X$ are  arbitrary functions and $\varepsilon$ is a real constant (not to be confused with the energy density $\epsilon$).
The only restrictions on these functions are that $\psi_X(t_1,r)=\psi_X(t_2,r)=0$,
and that they obey the same boundary conditions as imposed on the solutions to Eq.~\eqref{eq:wop} at the radial boundaries. We next define the function $s(\varepsilon)\equiv S[\boldsymbol\phi+\varepsilon\,\boldsymbol\psi,\dot{\boldsymbol\phi}+\varepsilon\,\dot{\boldsymbol\psi}]$, and, 
for simplicity, we assume that $(\varepsilon,\phi_X,\psi_X)$ are all real\footnote{We note that the complex case can also be easily obtained by assuming $\phi_X,\bar\phi_X$ as independent functions.}.~We obtain the equations of motion from $(ds/d\varepsilon)|_{\varepsilon = 0}=0$.

To simplify the calculation of $ds/d\varepsilon$, we separate the function $s(\varepsilon)$ into the following terms
\begin{align}
    s(\varepsilon) = \sum_X a_X(\varepsilon) - c_X(\varepsilon) - d_X(\varepsilon)\,,
\end{align}
with
\begin{subequations}\label{eq:helper}
\begin{gather}
    a_X(\varepsilon)=\frac12\int_{t_1}^{t_2} dt\, \langle \dot\phi_X+\varepsilon \, \dot\psi_X,\hat A_X(\dot\phi_X+\varepsilon\, \dot\psi_X)\rangle \ ,\\
    c_X(\varepsilon)=\frac12\int_{t_1}^{t_2} dt\, \langle \phi_X+\varepsilon\, \psi_X,\hat C_X(\phi_X+\, \varepsilon\psi_X)\rangle \ ,\\
    d_X(\varepsilon)=\frac12\int_{t_1}^{t_2} dt\, \langle \phi_X+\varepsilon\, \psi_X,\hat D_X(\phi_Y+\varepsilon\, \psi_Y)\rangle \ .
\end{gather}
\end{subequations}
Using the definition of the inner product, and integrating the derivative of $a_X$ by parts, we then find
\begin{subequations}
\begin{align}
    \frac{da_X}{d\varepsilon} &= \frac{1}{2} \int_{t_1}^{t_2} dt \left< \dot{\psi}_X,\hat A_X \dot{\phi}_X\right> + \left<\dot{\phi}_X,\hat A_X \dot{\psi}_X\right> + {\cal{O}}(\varepsilon)\,,
    \nonumber \\
    &= \int_{t_1}^{t_2} dt \int_0^{R_2} dr \; \dot{\psi}_X \hat A_X \dot{\phi}_X + {\cal{O}}(\varepsilon)\,,
    \label{eq:derive1}
    \\
    \frac{dc_X}{d\varepsilon} &=\frac{1}{2} \int_{t_1}^{t_2} dt \left< {\psi}_X,\hat C_X {\phi}_X\right> + \left<{\phi}_X,\hat C_X {\psi}_X\right> + {\cal{O}}(\varepsilon)\,,
    \nonumber \\
    &= \int_{t_1}^{t_2} dt \int_0^{R_2} dr \; {\psi}_X \hat C_X {\phi}_X + {\cal{O}}(\varepsilon)\,,
    \label{eq:derive2}
    \\
    \frac{dd_X}{d\varepsilon} &= \frac{1}{2} \int_{t_1}^{t_2} dt \left< {\psi}_X,\hat D_X {\phi}_Y\right> + \left<{\phi}_X,\hat D_X {\psi}_Y\right> + {\cal{O}}(\varepsilon)\,,
    \nonumber \\
    &= \frac{1}{2} \int_{t_1}^{t_2} dt \int_0^{R_2} dr \; \left({\psi}_X \hat D_X {\phi}_Y + {\psi}_Y \hat D_Y {\phi}_X \right) + {\cal{O}}(\varepsilon)\,,
    \label{eq:derive3}
\end{align}
\end{subequations}
where in the last line we used Eq.~\eqref{eq:sym}.
Thus, taking the limit as $\varepsilon$ goes to zero, we have
\begin{equation}\label{eq:variation}
    \left(\frac{d s}{d\varepsilon}\right)_{\varepsilon=0} \!\!\!\! \!\!=-\int_{t_1}^{t_2} \!\!\!\! dt\, \int_0^{R_2} \!\!\!\! dr\, \sum_X \psi_X(A_X\ddot\phi_X+C_X\phi_X+D_X\phi_Y).
\end{equation}
Setting $(ds/d\varepsilon)_{\varepsilon=0}=0$ for arbitrary $\psi_X$ implies that $\phi_X$ must satisfy Eq.~\eqref{eq:op}. 

We next define a Rayleigh quotient with the \textit{time-independent} displacements $\zeta_X(r)$,
in the same way as done in Appendix B of \cite{HTWW1965}.
Dividing the potential energy by the kinetic energy evaluated at $\zeta$ instead of $\dot\phi$ (see first equality in Eq.~\eqref{eq:action}), we then define
\begin{equation}\label{eq:2Rayleigh}
    I[\boldsymbol\zeta]\!\!=\!\!\frac{\langle \zeta_1,\hat C_1\zeta_1\rangle+\langle \zeta_2,\hat C_2\zeta_2\rangle+\langle \zeta_1,\hat D_1\zeta_2\rangle+\langle \zeta_2,\hat D_2\zeta_1\rangle}{\langle \zeta_1,\hat A_1\zeta_1\rangle+\langle\zeta_2,\hat A_2\zeta_2\rangle}.
\end{equation}
The quotient of Eq.~\eqref{eq:2Rayleigh} allows us to define a variational principle on the eigenfrequencies of the radial modes, which we express in the following theorem.

\begin{theorem}\label{th:vari}
    The eigenvalues $\omega^2$ of Eq.~\eqref{eq:2sl} with the boundary conditions of Eqs.~\eqref{eq:zero_boundary_2f} and \eqref{eq:outer_boundary_2f},
    are given by the extremal values of the Rayleigh quotient of Eq.~\eqref{eq:2Rayleigh}. The converse also holds: the extrema of said Rayleigh quotient are eigenvalues of the boundary value problem.
\end{theorem}

\begin{proof}
    We define the following function $i(\varepsilon)$ 
    \begin{equation}\label{eq:rayleigh_varied}
        i(\varepsilon)\equiv I[\boldsymbol\zeta+\varepsilon \, \boldsymbol\eta] \ ,
    \end{equation}
    where we choose $\boldsymbol{\zeta}$ such that it extremizes the Rayliegh quotient, i.e. $i'(0)=0$ and we once again assume that $\varepsilon$, $\zeta_X$, and $\eta_X$ are real for simplicity.
    We define the components $(\Tilde a_X,\Tilde c_X,\Tilde d_X)$ in the same way as the $(a_X,c_X,d_X)$ functions $\varepsilon$ in the Lagrangian [see Eq.~\eqref{eq:helper}], but we set $\phi_X\to\zeta_X$, $\dot{\phi}_X\to\zeta_X$, $\psi_X\to\eta_X$, and $\dot{\psi}_X\to\eta_X$, without integrating over time. 
    Then, we rewrite Eq.~\eqref{eq:rayleigh_varied} as
    \begin{equation}
        i(\varepsilon)=\frac{\sum_X \Tilde c_X(\varepsilon)+\Tilde d_X(\varepsilon)}{\sum_X \Tilde a_X(\varepsilon)} \ .
    \end{equation}

    As before, our ultimate aim is to compute $(di/d\varepsilon)_{\varepsilon=0}=0$. 
    We first take the derivative to find
    \begin{equation}
    \begin{split}
        \frac{di}{d\varepsilon} &=\frac{\sum_X (d\Tilde c_X/d\varepsilon)+(d\Tilde d_X/d\varepsilon)}{\sum_X \Tilde a_X(\varepsilon)}
         \\
        &-\frac{\sum_X (d\Tilde a_X/d\varepsilon)}{\sum_X \Tilde a_X(\varepsilon)}\frac{\sum_X \Tilde c_X(\varepsilon)+\Tilde d_X(\varepsilon)}{\sum_X\Tilde a_X(\varepsilon)}\\
        &=\frac{\sum_X (d\Tilde c_X/\varepsilon)+(d\Tilde d_X/d\varepsilon)-i(\varepsilon) (d\Tilde a_X/d\varepsilon)}{\sum_X \Tilde a_X(\varepsilon)} \ .
    \end{split}
    \end{equation}
    Using the derivatives in Eqs.~\eqref{eq:derive1}-\eqref{eq:derive3}, and setting $\varepsilon=0$, gives us
    \begin{equation}
    \label{eq:linear-expansion-rayleigh-quotient}
        0=\sum_X\int_0^{R_2} dr\, \eta_X\left[\hat C_X\zeta_X+\hat D_X\zeta_Y-i(0)\hat A_X\zeta_X\right] \ .
    \end{equation}
    Once again as $\eta_X$ is arbitrary, the bracketed term in Eq.~\eqref{eq:linear-expansion-rayleigh-quotient} must be identically zero.

    This results in Eq.~\eqref{eq:2sl} but with $i(0)$ replacing $\omega^2$.
    We conclude that $\boldsymbol\zeta$ is a solution to Eq.~\eqref{eq:2sl} under boundary conditions Eqs.~\eqref{eq:zero_boundary_2f} and \eqref{eq:outer_boundary_2f} with eigenvalue $i(0)$. Thus, the extremal values of the Rayliegh quotient are eigenvalues.
    
    Now suppose that $\omega^2$ is an eigenvalue of the boundary value problem Eqs.~\eqref{eq:2sl}, \eqref{eq:zero_boundary_2f} and \eqref{eq:outer_boundary_2f} with corresponding eigenvector $\boldsymbol\zeta$.
    Then define $i_{\boldsymbol\zeta}(\varepsilon)$ in the same way as in Eq.~\eqref{eq:rayleigh_varied} for an arbitrary $\boldsymbol{\eta}$.
    It follows that $I[\boldsymbol{\zeta}]=i(0)=\omega^2$.
    Then, by the same mathematical procedure as above, it follows that
    \begin{equation}
        \left(\frac{di_{\boldsymbol\zeta}}{d\varepsilon}\right)_{\varepsilon=0}=\sum_X\int_0^{R_2}dr\, \eta_X\left[\hat C_X\zeta_X+\hat D_X\zeta_Y-i_{\boldsymbol{\zeta}}(0)\hat A_X\zeta_X\right]\ ,
    \end{equation}
    and since $i_{\boldsymbol{\zeta}}(0)=\omega^2$ is an eigenvalue with eigenvector $\boldsymbol\zeta$, the right-hand side is exactly zero.
    Therefore, $(di_{\boldsymbol{\zeta}}/d\varepsilon)_{\varepsilon=0}=0$, and thus, $\boldsymbol\zeta$ extremizes the Rayliegh quotient, and since $\omega^2=I[\boldsymbol\zeta]$, eigenvalues are extrema of the Rayliegh quotient.
    
\end{proof}

\subsection{Symmetric Operator}\label{ssec:operator}


We next show that the eigenvalue equation for two fluids [Eq.~\eqref{eq:2sl}] is equivalent to an eigenvalue problem for a self-adjoint operator in a well-chosen Hilbert space.
From this, we conclude that all modes are real, and that the associated eigenvectors form a complete basis. 

To show this, we work with almost the same Hilbert space used to prove completeness for single-fluid radial perturbations.
More precisely, for each fluid perturbation we define a Hilbert space $H_X$, which is the space of square integrable functions weighted by $A_X$ on the interval $(0,R_X)$. 
The total Hilbert space is the direct sum of each individual Hilbert space: $H=H_1\oplus H_2$,
where we have that $\zeta_1\in H_1$ and $\zeta_2\in H_2$. The inner product between the two vectors $\boldsymbol\zeta=(\zeta_1,\zeta_2)$ and $\boldsymbol\eta=(\eta_1,\eta_2)$ is then
\begin{equation}
    \label{eq:sturm-liouville-inner-product-space-two-fluid} 
    \braket{\boldsymbol\eta}{\boldsymbol\zeta}=\int_0^{R_1}dr\, \bar\eta_1\hat A_1\zeta_1+\int_0^{R_2}dr\, \bar\eta_2\hat A_2\zeta_2 \ .
\end{equation}
For notational simplicity, we write 
Eq.~\eqref{eq:sturm-liouville-inner-product-space-two-fluid} as follows
\begin{equation}
\label{eq:simplified-notation-two-fluid-hilbert-space}
\braket{\boldsymbol\eta}{\boldsymbol\zeta}=\int dr\, \begin{pmatrix}\bar\eta_1&\bar\eta_2\end{pmatrix}\begin{pmatrix}\hat A_1& 0\\0&\hat A_2\end{pmatrix}\begin{pmatrix}\zeta_1\\\zeta_2\end{pmatrix} \ ,
\end{equation}
where the linear algebra operations above are to be done with a flat internal space metric.
In Eq.~\eqref{eq:simplified-notation-two-fluid-hilbert-space}, the radial integral is taken in $(0,R_1)$ or in $(0,R_2)$ depending on the context. Note that this inner product is \textit{different} from that defined in Eq.~\eqref{eq:unweighed-inner-product}, which was used for single fluids and had no weights. 

Just as we did for a single fluid, we define a potential energy operator for the Hilbert space
\begin{equation}\label{eq:potoperator}
    \mathcal V =\begin{pmatrix}\hat A_1^{-1}\hat C_1&\hat A_1^{-1}\hat D_1\\\hat A_2^{-1}\hat D_2&\hat A_2^{-1}\hat C_2\end{pmatrix} \ .
\end{equation}
This definition, along with the inner product of Eq.~\eqref{eq:simplified-notation-two-fluid-hilbert-space}, allows us to write the potential energy of radial perturbations as
\begin{equation}
    V[\boldsymbol\zeta]=\frac12\braket{\boldsymbol\zeta}{\mathcal V\boldsymbol\zeta} \ .
\end{equation}
We also note that the operator $\mathcal V$ is symmetric on $H$, namely
\begin{equation}
\begin{split}
    \braket{\mathcal V\boldsymbol\eta}{\boldsymbol\zeta}=&\langle \hat A_1^{-1}\hat C_1\eta_1,\hat A_1\zeta_1\rangle+\langle \hat A_2^{-1}\hat C_2\eta_2,\hat A_2\zeta_2\rangle\\
    &+\langle \hat A_1^{-1}\hat D_1\eta_2,\hat A_1\zeta_1\rangle+\langle \hat A_2^{-1}\hat D_2\eta_1,\hat A_2\zeta_2\rangle\\
    =&\langle \eta_1,\hat C_1\zeta_1\rangle+\langle\eta_2,\hat C_2\zeta_2\rangle+\langle \eta_2,\hat D_2\zeta_1\rangle+\langle \eta_1,\hat D_1\zeta_2\rangle\\
    =&\braket{\boldsymbol\eta}{\mathcal V\boldsymbol\zeta} \ ,
\end{split}
\end{equation}
which follows from $\hat A_X$ and $\hat C_X$ being symmetric operators, and from the relation between the $\hat D_X$ operators in Eq.~\eqref{eq:sym}
The Rayleigh quotient [Eq.~\eqref{eq:2Rayleigh}] is then
\begin{equation}\label{eq:2action}
    I[\boldsymbol\zeta]=\frac{\braket{\boldsymbol\zeta}{\mathcal V\boldsymbol\zeta}}{\braket{\boldsymbol\zeta}} \ .
\end{equation}
Finally, we write the equations of motion and the mode equations [Eqs.~\eqref{eq:op} and \eqref{eq:2sl}] in this notation as
\begin{gather}
    \mathbf {\hat A}(\partial_t^2+\mathcal V)\boldsymbol\zeta=0 \ , \\
    \mathbf {\hat A}(\mathcal V-\omega^2\mathbbm1)\boldsymbol\zeta=0 \ ,
\end{gather}
where $\mathbf {\hat A}$ is the diagonal matrix $\operatorname{diag}(\hat A_1,\hat A_2)$ that defines the inner product in Eq.~\eqref{eq:simplified-notation-two-fluid-hilbert-space}.

While we have shown that $\mathcal{V}$ is symmetric, we have not yet shown that it is self-adjoint, which is a crucial property to prove the completeness of modes. Recall that an operator is symmetric provided that $\braket{\mathcal V\boldsymbol\eta}{\boldsymbol\zeta}=\braket{\boldsymbol\eta}{\mathcal V\boldsymbol\zeta}$ for all vectors $\boldsymbol\eta,\boldsymbol\zeta$ in the domain of $\mathcal V$, but for the operator to also be self-adjoint, one must additionally have that the domain of the adjoint of $\mathcal V$ be the same as the domain of $\mathcal V$.
For example, the operator $id/dr$ in the space of square integrable functions that are zero at both boundaries is symmetric. 
However, the domain of its adjoint has no conditions at the boundaries, and thus, the adjoint has a larger domain, meaning that $id/dr$ is not self-adjoint in such a space.
For further reading, we refer the audience to Sec. 1(c) of \cite{DysonSchutz1979}, which follows Chapter $8$, Section 119 of \cite{RieszNagy1956}, from where we obtained the exampled outlined in the previous sentence.
Fortunately, we can show that $\mathcal{V}$ is practically self-adjoint by constructing a canonical self-adjoint extension to the operator, via the Stone-Friedrichs theorem, similar to what is done in \cite{DysonSchutz1979}\footnote{It is possible that $\mathcal V$ itself is self-adjoint, or that at the very least its closure is self-adjoint, but we will not consider here an analysis of the domains of the operators.}.
The Stone-Friedrichs theorem states that every semi-bounded symmetric operator can be extended to a semi-bounded self-adjoint operator with the same bound \cite{RieszNagy1956}.
Thus, by proving that $\mathcal V$ is bounded from below, we can use the self-adjoint Friedrich extension of $\mathcal{V}$ to prove the completeness of eigenmodes.



Before continuing, we provide a brief overview of our argument. 
If the expression in Eq.~\eqref{eq:2action} is bounded from below, then there exists a self-adjoint extension with the same lower bound \cite{RieszNagy1956}.
We can then use that extension for any statement that requires a self-adjoint operator.
In particular, this shows that the equations of motion [Eq.~\eqref{eq:op}] are self-adjoint, implying that the mode solutions are complete, and the eigenvalues are real. 
Thus, we only need to analyze the mode solutions to establish radial stability.

\subsection{Lower bound}\label{ssec:low}

We here calculate an explicit lower bound to the Rayleigh quotient $I$ of Eq.~\eqref{eq:2Rayleigh},
in an analogous way to what we did for a single fluid in Sec.~\ref{ssec:bound}.
For an unbounded self-adjoint operator, it is possible for the spectrum to include a continuous part.
Within physics, this is seen with the scattering states of the hydrogen atom, which must be integrated over.
In principle, the continuous spectrum can have values lower than any eigenvalue, but not below the lower bound of the Rayleigh quotient.
Importantly, stability is determined by the smallest element of the spectrum of (the Friedrich extension of) $\mathcal V$, be it an eigenvalue or a continuous mode.
However, calculating the continuous modes is more difficult than calculating the eigenvalues, and, hence, we would like our statements to be in terms of the latter.
Moreover, we do not expect there to be a continuous spectrum since a continuous part does not exist for radial modes of single-fluid stars.
Therefore, we propose (and will later prove) the following stability theorem in terms of only eigenvalues, by assuming that the smallest value of the continuous spectrum is greater than the smallest eigenvalue (once the latter is shown to exist).


\begin{theorem}\label{th:fundamental}
    Let $c_{s,1}^2$ and $c_{s,2}^2$ be the (squares of the) speeds of sound for each of the fluids inside a two-component star.
    Suppose that both speeds of sound are differentiable with respect to the radial coordinate on the interval $(0,R_1)$, and non-zero inside the inner boundary of the star [i.e.~$r\in(0,R_1)$ with the assumption $R_1\leq R_2$].
    Then, there exists a fundamental mode with frequency $\omega_0^2$ that solves Eq.~\eqref{eq:2sl}, such that for all other eigenfrequencies $\omega_n^2$ with solutions to Eq.~\eqref{eq:2sl}, satisfy $\omega_0^2<\omega_n^2$.
    Furthermore, let us assume that this fundamental mode is lesser than the smallest value of the continuous spectrum of $\mathcal V$, if the latter exists.
    An equilibrium configuration will then be stable if and only if its corresponding fundamental mode $\omega_0^2>0$.
\end{theorem}

\begin{proof}
Suppose we have an equilibrium configuration as described in Sec.~\ref{ssec:back2}.
By the background equations [Eqs.~\eqref{eq:px} and \eqref{eq:2phi2}], both the energy density and pressure for both fluids is maximal at the center of the star.
An equilibrium configuration is determined by specifying an equation of state, along with the central energy densities $\epsilon_X^c$ for each fluid with $X \in (1,2)$.
Without loss of generality, we assume that $R_2\geq R_1$.

We want to prove that $I[\boldsymbol\zeta]$, defined in Eq.~\eqref{eq:2Rayleigh}, is bounded from below by some positive constant for all vectors $\boldsymbol\zeta=(\zeta_1,\zeta_2)$ in the domain of $\mathcal V$, so we start by recasting $I$ in a more explicit form.
We restrict the proof to the case where $\boldsymbol\zeta$ is real valued\footnote{We note that while we could in general assume that $\zeta$ is complex, the algebra is simplified in the interaction terms by assuming the real case.}, and we assume $\braket{\boldsymbol\zeta}=1$ without loss of generality.
We then define
\begin{align}\label{eq:intex}
    \mathcal I_X
    &=
    \int_0^{R_X}dr\, A_X(\zeta_X)^2
    =
    \int_0^{R_X}e^{\Lambda-3\Phi}\frac{(\epsilon_X+p_X)}{r^2}(\zeta_X)^2
    ,
\end{align}
so that setting $\braket{\boldsymbol\zeta}=1$ is equivalent to setting $\mathcal I_1+\mathcal I_2=1$.
Expanding all terms in Eq.~\eqref{eq:2Rayleigh}, and using the definition for the operators in Eq.~\eqref{eq:dx}, we then have
\begin{align}\label{eq:tobound}
    I[\boldsymbol\zeta]
    =\int dr\, \Bigg[
            &P_1\left(\frac{d}{dr}\zeta_1\right)^2+(Q_1+R)(\zeta_1)^2+P_2\left(\frac{d}{dr}\zeta_2\right)^2
            \nonumber\\
            &+(Q_2+R)(\zeta_2)^2+2S\zeta_1\zeta_2
            \nonumber\\
            &+2f_1\left(\frac{d}{dr}\zeta_1\right)\zeta_2+2f_2\zeta_1\left(\frac{d}{dr}\zeta_2\right)
        \Bigg] \ .
\end{align}
We have not explicitly written out the bounds of integration for notational simplicity, but the choice of bounds is easy to infer from the integrands.
Concretely, terms that only contain powers of $\zeta_1$, such as $(\zeta_1)^2$, are integrated inside $(0,R_1)$, while terms that contain only powers of $\zeta_2$ are integrated inside $(0,R_2)$.
Terms that multiply $\zeta_1$ and $\zeta_2$ are integrated over $(0,R_1)$, as $\zeta_1=0$ when $r>R_1$.

We can decompose $I[\boldsymbol\zeta]$ further into terms that are not explicitly positive definite and terms that are. By using the definitions of the functions $[P_X(r),Q_X(r),f_X(r),R(r),S(r)]$ in Eqs.~\eqref{eq:2funcs} inside of Eq.~\eqref{eq:intex}, we then find
\begin{equation}\label{eq:sumall}
    I[\boldsymbol\zeta] = a_1+a_2+a_3+a_4 + \delta a \ ,
\end{equation}
where we have defined
\begin{widetext}
\begin{subequations}
\begin{align}
    \label{eq:a1}
    a_1
    \equiv&
    \int dr\, \frac{e^{\Lambda-3\Phi}}{r^2}\left[(\Phi')^2-r^2\frac{d}{dr}\left(\frac{\Phi'}{r^2}\right)\right]\left[(\epsilon_1+p_1)(\zeta_1)^2+(\epsilon_2+p_2)(\zeta_2)^2\right]
    ,\\
    \label{eq:a2}
    a_2
    \equiv&
   \int dr\, \frac{4\pi}{r}e^{3\Lambda-3\Phi}\left(\Phi'-\frac1r\right)\left[(\epsilon_1+p_1)^2(\zeta_1)^2+(\epsilon_2+p_2)^2(\zeta_2)^2\right] 
    ,\\
    a_3
    \equiv&
    \int dr\, 16\pi^2e^{5\Lambda-3\Phi}\left[\gamma_1p_1(\epsilon_2+p_2)^2+\gamma_2p_2(\epsilon_1+p_1)^2\right](\zeta_1-\zeta_2)^2 \nonumber\\
    &\quad +4\pi \frac{e^{3\Lambda-3\Phi}}{r}(\epsilon_1+p_1)(\epsilon_2+p_2)\left\{\left[(\zeta_1)^2+(\zeta_2)^2\right]\Phi'-2\zeta_1\zeta_2\frac1r\right\}
    , \label{eq:a3}\\
    a_4
    \equiv&
    2\int dr\, (\zeta_2-\zeta_1)\left(f_1\frac{d}{dr}\zeta_1-f_2\frac{d}{dr}\zeta_2\right)
    \label{eq:a_4}.
\end{align}
\end{subequations}
\end{widetext}
Looking at Eq.~\eqref{eq:tobound}, the $a_1$ term comes from the $Q_X$ terms, specifically those from the first line in Eq.~\eqref{eq:2Qb}.
The $a_2$ term also comes from the $Q_X$ terms, but from the second line of Eq.~\eqref{eq:2Qb} without the $\frac{d}{dr}f_X$ term in such line.
The term $a_3$ comes from combining the terms with $S$ and $R$ in Eq.~\eqref{eq:tobound}, according to Eqs.~\eqref{eq:Rb} and \eqref{eq:Sb}.
Finally, the $a_4$ term comes from the two terms in the last line of Eq.~\eqref{eq:tobound}, added to the component $\frac{d}{dr}f_X$ in the definition $Q_X$ [see Eq.~\eqref{eq:2Qb}], and integrating by parts to simplify.


The $\delta a$ term in Eq.~\eqref{eq:sumall} stands for the terms that are obviously positive definite. More concretely, these terms are given by
\begin{align}
    \delta a &= \int dr\, P_1\left(\frac{d}{dr}\zeta_1\right)^2+P_2\left(\frac{d}{dr}\zeta_2\right)^2\, .
\end{align}
We have not included these terms in the definitions of $a_i$ because they are not needed to prove the theorem we are concerned with here. More concretely, the functions $P_X$ are positive definite, and hence, their integral multiplying $(\frac{d}{dr}\zeta_X)^2$ is also positive definite. Therefore, if we can show that $a_1$, $a_2+a_3$, and $a_4$ are all bounded from below, then so will be $I[\boldsymbol\zeta]$ regardless of these other positive-definite terms, which only change the value of the minimum, i.e.~if $a_1 + a_2 + a_3 + a_4 \geq a$ with $a \in \mathbb{R}$, then $I[\boldsymbol\zeta] \geq a$ because we can write $\delta a = (b)^2$ for some $b \in \mathbb{R}$.


%
Our goal is to place a lower bound on $I\left[\boldsymbol\zeta\right]$, and we will do so by showing that each of the $a_1$, $a_2+a_3$, and $a_4$ terms is bounded from below by a constant, which happens to multiply $\mathcal I_1+\mathcal I_2$.
%
%
%
To bound the individual terms in any of the $a_i$, we go through a similar series of steps as we did for a single-fluid star
(see Sec~\ref{ssec:bound}, and especially, Eqs.~\eqref{eq:onebound}-\eqref{eq:bound} and the surrounding text). 
We first consider the bound on $a_1$.
First, looking at Eq.~\eqref{eq:a1}, the term proportional to $(\Phi')^2$ is positive definite, and thus, it is bounded from below by zero, leading to
\begin{equation}
\begin{split}
    a_1
    \geq
    -\int dr\, &\frac{e^{\Lambda-3\Phi}}{r^2}\left[r^2\frac{d}{dr}\left(\frac{\Phi'}{r^2}\right)\right]\\
    &\times\left[(\epsilon_1+p_1)(\zeta_1)^2+(\epsilon_2+p_2)(\zeta_2)^2\right]
    .
\end{split}
\end{equation}
We now focus on the term proportional to $(\epsilon_1+p_1)(\zeta_1)^2$, from which we obtain
\begin{equation}\label{eq:suprem}
\begin{split}
    -&\int dr\, \frac{e^{\Lambda-3\Phi}}{r^2}\left[r^2\frac{d}{dr}\left(\frac{\Phi'}{r^2}\right)\right](\epsilon_1+p_1)(\zeta_1)^2\\
    \geq-& \sup\left[r^2\frac{d}{dr}\left(\frac{\Phi'}{r^2}\right)\right]\int dr\, e^{\Lambda-3\Phi}\frac{(\epsilon_1+p_1)}{r^2}(\zeta_1)^2 \ ,
\end{split}
\end{equation}
using the same arguments as done below Eq.~\eqref{eq:onebound}. 
The leftover integral in Eq.~\eqref{eq:suprem} is just $\mathcal I_1$ from Eq.~\eqref{eq:intex}.
Doing the same for the term proportional to  $(\epsilon_2+p_2)(\zeta_2)^2$ in Eq.~\eqref{eq:a1} and adding the results together, we have that
\begin{equation}\label{eq:a1_result}
    a_1\geq -\sup\left[r^2\frac{d}{dr}\left(\frac{\Phi'}{r^2}\right)\right](\mathcal I_1+\mathcal I_2) \ .
\end{equation}
We note that the supremum in Eq.~\eqref{eq:suprem} is on the interval $(0,R_1)$, and the equivalent expression for the second fluid will have the supremum defined on $(0,R_2)$.
Therefore, to obtain the factorization in Eq.~\eqref{eq:a1_result} we must take the larger of the two supremums for the inequality to work 
--and since $R_2\geq R_1$, the supremum in Eq.~\eqref{eq:a1_result} is on $(0,R_2)$.
All the supremums in the expressions for $a_1, a_2$, and $a_3$ are also on the interval $(0,R_2)$.

We next consider the bound on $a_2$.
We follow the same steps as for $a_1$ by taking the supremum of the terms not inside the integrand of $\mathcal I_X$ and factoring them out. More concretely, the term proportional to $(\zeta_1)^2$ in Eq.~\eqref{eq:a2} can be factorized as
\begin{align}
    &- \int dr\, \frac{e^{\Lambda-3\Phi}}{r^2} \left[\left(4 \pi e^{2\Lambda}\right) \left(1-r \Phi'\right) (\epsilon_1+p_1) \right]
    \nonumber \\ 
    &\qquad \qquad \qquad \times\left[(\epsilon_1+p_1)(\zeta_1)^2\right]
    \nonumber \\
    &\geq - {\rm sup} \left[\left(4 \pi r e^{2\Lambda}\right) \left(\frac{1}{r}- \Phi'\right) (\epsilon_1+p_1) \right] {\cal{I}}_1
    \nonumber \\
    &\geq - {\rm sup} \left[\left(4 \pi r e^{2\Lambda}\right) \left(\frac{1}{r}- \Phi'\right)  \right] \rho_1^{\rm max} {\cal{I}}_1 \,,
\end{align}
where the second inequality comes from the fact that $-{\rm sup}(a \, b) \geq -{\rm sup}(a) \, {\rm sup}(b)$ for all $(a,b) \in \mathbb{R}_{>0}$, and where $\rho_X^{\rm max}\equiv(\epsilon_X^c+p_X^c)$ is a constant for any specified equilibrium configuration.
As a result, we obtain the following bound (equivalent to Eq.~\eqref{eq:onebound} for a single fluid)
\begin{equation}\label{eq:bounda2}
    a_2\geq -\sup\left[4\pi re^{2\Lambda}\left(\frac1r - \Phi'\right)\right](\rho_1^{\rm max}\mathcal I_1+\rho_2^{\rm max}\mathcal I_2) \ .
\end{equation}
This expression for $a_2$ is not yet directly proportional to $\mathcal I_1+\mathcal I_2$, but we will show that we can combine it with the bound on $a_3$ to obtain an expression that is.

Let us then move on to bounding $a_3$.
We start by noticing that the first line of Eq.~\eqref{eq:a3} is positive definite and thus bounded from below by zero.
Next, we note that $-2\zeta_1\zeta_2=(\zeta_1-\zeta_2)^2-[(\zeta_1)^2+(\zeta_2)^2]$ so that we obtain
\begin{equation}
\begin{split}
    &\int dr\, 4\pi \frac{e^{3\Lambda-3\Phi}}{r^2}(\epsilon_1+p_1)(\epsilon_2+p_2)(-2\zeta_1\zeta_2)\\
    =&\int dr\, 4\pi \frac{e^{3\Lambda-3\Phi}}{r^2}(\epsilon_1+p_1)(\epsilon_2+p_2)\\
    &\qquad \times\left\{(\zeta_1-\zeta_2)^2-\left[(\zeta_1)^2+(\zeta_2)^2\right]\right\} \ .
\end{split}
\end{equation}
Once again,  term in the $(\zeta_1-\zeta_1)^2$ term, and thus, this term is bounded from below by zero.
After carrying out this substitution in Eq.~\eqref{eq:a3} and removing the terms bounded from below by zero, we obtain the following expression
\begin{align}
\label{eq:a3_sec}
    a_3
    &\geq \int dr\, 4\pi \frac{e^{3\Lambda-3\Phi}}{r}(\epsilon_1+p_1)(\epsilon_2+p_2)
    \nonumber \\
    &\times \left[(\zeta_1)^2+(\zeta_2)^2\right]\left(\Phi'-\frac1r\right) \ .
\end{align}
This result is similar to the expression for $a_2$ in Eq.~\eqref{eq:a2}, but with $(\epsilon_1+p_1)^2(\zeta_1)^2\to (\epsilon_1+p_1)(\epsilon_2+p_2)(\zeta_1)^2$, and the same under index interchange $1 \leftrightarrow 2$. 
Therefore, we obtain a similar bound as before, namely
\begin{equation}
    a_3\geq -\sup\left[4\pi re^{2\Lambda}\left(\frac1r-\Phi'\right)\right]\left(\rho_2^{\rm max}\mathcal I_1+\rho_1^{\rm max}\mathcal I_2\right) \ .
\end{equation}
This bound for $a_3$ is not directly proportional to $\mathcal I_1+\mathcal I_2$ either, but when we add it with the bound for $a_2$ in Eq.~\eqref{eq:bounda2}
the result is directly proportional to $\mathcal I_1+\mathcal I_2$. 

We now look at the bound we have derived so far for the sum of $a_1, a_2$, and $a_3$, namely
\begin{equation}\label{eq:a123inequality}
\begin{split}
    a_1+a_2+a_3\geq &-\sup_{r\in(0,R_2)}\left[r^2\frac{d}{dr}\left(\frac{\Phi'}{r^2}\right)\right]\\
    &-4\pi(\rho_1^{\rm max}+\rho_2^{\rm max})\sup_{r\in(0,R_2)}\left[ e^{2\Lambda}(1-r\Phi')\right] \ ,
\end{split}
\end{equation}
where we have used that $\mathcal I_1+\mathcal I_2=1$. 
Equation~\eqref{eq:a123inequality} is the same as the one-fluid bound in Eq.~\eqref{eq:bound}, with $\rho^{\rm max}\to \rho_1^{\rm max}+\rho_2^{\rm max}$.
Therefore, just as in the case of a single-fluid, $a_1 + a_2 + a_3$ is bounded from below. 
Clearly, however, the above inequality does not yet bound $I[\zeta]$, since we must also study the $a_4$ term, which importantly includes interactions between the two fluids. 

Let us then try to place a bound on $a_4$.
The $a_4$ term requires more work to bound, and the process of doing so is lengthy and unilluminating, which is why we have saved it for Appendix \ref{app:lastbound}. In the end and after much work, we obtain
\allowdisplaybreaks[4] 
\begin{widetext}
\begin{equation}\label{eq:a4}
    \begin{split}
        a_4\geq &-4\pi (\rho_1^{\rm max}+\rho_2^{\rm max})\left\{\sup_{r\in(0,R_1)}\left[e^{2\Lambda}(1-2r\Phi')\right]+\frac12\sup_{r\in(0,R_1)}\left[-e^{2\Lambda}r\Phi'\left(2+\frac{c_{s,1}^2}{c_{s,2}^2}+\frac{c_{s,2}^2}{c_{s,1}^2}\right)\right]\right\}\\
        &-8\pi(\rho_1^{\rm max}+\rho_2^{\rm max})\left(\alpha_1+\alpha_2\right)+2\pi(\rho_1^{\rm max}+\rho_2^{\rm max})\left(\beta_1+\beta_2\right) \\ 
        &-96\pi^2(\rho_1^{\rm max}+\rho_2^{\rm max})^2\sup_{r\in(0,R_1)}\left(r^2e^{4\Lambda}\right) \ ,
    \end{split}
\end{equation}
\end{widetext}
where the constants $\alpha_X$ and $\beta_X$ are given by
\begin{subequations}
\begin{align}
    \alpha_X&=\max\left\{0,\sup_{r\in(0,R_1)}\left(re^{2\Lambda}\frac{d}{dr}c_{s,X}^2\right)\right\} \ , \label{eq:alphax}\\
    \beta_X&=\min\left\{0,\inf_{r\in(0,R_1)}\left(re^{2\Lambda}\frac{d}{dr}c_{s,X}^2\right)\right\} \label{eq:betax} \ ,
\end{align}
\end{subequations}
and where we recall that $c_X$ is the speed of sound for fluid $X \in (1,2)$.
We note that unlike the expressions for $a_1, a_2$, and $a_3$, all supremums and infinums in the expression for $a_4$ are on the interval $(0,R_1)$, i.e. within the inner boundary.
Provided the assumptions of the theorem hold (i.e.~that $(\alpha_X,\beta_X)$ is bounded from below and $c_{s,X}^2 \neq 0$), then $a_4$ is also bounded from below.

Before continuing, we discuss several implications of the bound on $a_4$ [Eq.~\eqref{eq:a4}].
First, as we stated previously, taking the one fluid limit for $a_1+a_2+a_3$ gives us the single-fluid bound Eq.~\eqref{eq:bound}.
If we take the one-fluid limit (i.e. removing the first fluid $\rho_1^{\rm max}\to 0,R_1\to0$) of the right-hand side of Eq.~\eqref{eq:a4}, we obtain a nonzero result (specifically, this limit returns $-4\pi\rho_2^{\rm max}$).
Therefore, adding the right-hand side of Eq.~\eqref{eq:a4} before taking the one-fluid limit results in a \textit{lower} bound for a single fluid than expected.
However, $a_4$ does still go to zero when taking the one-fluid limit.
This can be deduced from the results in Appendix \ref{app:lastbound}, where we show that $a_4$ is proportional to $(\rho_2^{\rm max}\mathcal{I}_1+\rho_1^{\rm max}\mathcal{I}_2)$, which does in fact vanish when taking the one-fluid limit. 

The existence of the fundamental mode then follows from the following argument. First, $I[\boldsymbol\zeta]$ as written in Eq.~\eqref{eq:sumall} is bounded from below, 
\begin{equation}\label{eq:fund1}
    \min_{\boldsymbol\zeta\in D(\mathcal V)}I[\boldsymbol\zeta]\geq a_1+a_2+a_3+a_4 \ ,
\end{equation}
because both $a_1 + a_2 + a_3$ and $a_4$ are bounded from below, as proven above.
Since the eigenvalues are in the image of the Rayleigh quotient (by Theorem~\ref{th:vari}), then the eigenvalues are also bounded from below, and the smallest of these eigenvalues is defined to be the fundamental mode.
Moreover, since the eigenvalues are extremal values, and the minimum over the domain is itself an extremum value, then the eigenvalue of the fundamental mode $\omega_0^2$ is just 
\begin{equation}\label{eq:fund2}
    \omega_0^2 = \min_{\boldsymbol\zeta\in D(\mathcal V)}I[\boldsymbol\zeta] \ .
\end{equation}
Putting these two results together, we then have that
\begin{equation}\label{eq:fund}
    \omega_0^2 = \min_{\boldsymbol\zeta\in D(\mathcal V)}I[\boldsymbol\zeta]\geq a_1+a_2+a_3+a_4 \ ,
\end{equation}
which means that fundamental frequency $\omega_0^2$ is bounded from below.

With this bound at hand, we can now arrive at the conclusion of our proof. As was discussed in Sec.~\ref{ssec:variation}, this bound allows us to define a canonical self-adjoint extension for $\mathcal V$, and thus, the mode solutions are complete.
Then, the general solution obtained from the modes will have a norm that grows over time if and only if some mode has a negative eigenvalue.
Since the fundamental mode is the smallest of all modes, given the assumptions stated in the theorem, then there can only exist a mode with a negative eigenvalue if the fundamental mode itself is negative.
Therefore, we conclude that an equilibrium configuration is stable if and only if its fundamental mode is greater than zero.


\end{proof}

The assumptions we made about the speeds of sound in the statement of Theorem \ref{th:fundamental} are necessary, so that $a_4$ is indeed bounded in Eq.~\eqref{eq:a4}.
We consider these assumptions to be physically reasonable.
First, we assumed  that $c_{s,X}^2$ is differentiable, 
which ensures that both $\inf(re^{2\Lambda}\frac{d}{dr}c_{s,X}^2)$ and $\sup(re^{2\Lambda}\frac{d}{dr}c_{s,X}^2)$ are finite.
Note that we allow for the speed of sound to change rapidly along the star, just as long as it is not infinitely fast.  
Second, we assumed that neither of the speeds of sounds can vanish inside $(0,R_1)$ to ensure that both $c_{s,1}^2/c_{s,2}^2$ and $c_{s,2}^2/c_{s,1}^2$ are finite within this interval.
The speed of sound goes to zero at first-order phase transitions \cite{Tan2022,Shahrbaf2020} and it is non-differentiable for higher-order phase transitions, which means that Theorem \ref{th:fundamental} only holds for equations of state without such phase transitions.

However, first- and higher-order phase transitions are already a problem when using the standard stability criterion for single-fluid stars.
Previous work has shown that the existence of first-order phase transitions for single-fluid stars
can lead to either additional stability branches \cite{Gerlach1968,Alford2017,Tan2022} or to a mass-radius maximum that does not coincide with a central density at which $\omega_0^2=0$ \cite{DiClemente2020,Vasquez-Flores2012,Pereira2018}.
Moreover, at any point $r^*$ where $c_{s,X}^2(r^*)=0$, the function $P_X(r^*)$ of the Sturm-Liouville operator will vanish.
This would cause the Sturm-Liouville operator (see Eq.~\eqref{eq:sl} for the one-fluid, and Eq.~\eqref{eq:2sl} for the two-fluid case) to not be well-defined at $r^*$.
Therefore, one already needs to adapt the mode equations (in this case, Eq. \eqref{eq:2sl}) appropriately to include correct interface conditions across the boundary of the phase transition \cite{Rau2023,Pereira2018,DiClemente2023}.
These first-order phase transitions have been studied in the context of neutron stars with a quark-matter core, and the ``misbehavior'' of the EoS occurs at the transition between the quark-matter core and the neutron matter.
Noting that Theorem \ref{th:fundamental} is valid as long as $c_{s,X}^2\to 0$ at radius greater than $R_1$, we conclude that, as long as all of the second fluid is contained in a radius smaller than the quark-matter transition radius, Theorem \ref{th:fundamental} will hold for such stars.

We emphasize that we likely did not find the ``best'' or most stringent lower bound possible, and we make no attempt to do so in this article.
In fact, from our calculations in Appendix \ref{app:lastbound},
the expression for $a_4$ in Eq.~\eqref{eq:a4} can be made higher (by for example, some terms having $-\sup(\frac{d}{dr}c_{s,1}^2+\frac{d}{dr}c_{s,2}^2)$ instead of the sum of the supremums). What we presented here, however, is a simpler way of writing the bound (as compared to Eq.~\eqref{eq:fulla4})  
Our goal in the section was to show that such a bound exists, which in turn proved that a fundamental mode exists.
Generally, we expect the actual value of the fundamental mode for an equilibrium configuration to be much higher than this bound,
and, in fact, since this bound is generally negative, any stable configuration requires for such mode to be higher than the bound found here.

\section{Radial stability of stars made up of an arbitrary number of non-interacting fluids}\label{sec:multi}


We here extend the results presented in Secs. \ref{sec:lang} and most of \ref{sec:energy} (but not those of Sec.~\ref{ssec:low}) to a star that is made up of an arbitrary number of non-interacting fluids.
Multifluid stars were previously considered in \cite{Kain2021}, where the equations of structure were derived for arbitrarily many fluids, but radial stability was not analyzed beyond numerical calculations for a simple three-fluid model.
The particular case of three fluids was additionally numerically studied in \cite{Rezaei2018}. 
As we have done before, we restrict attention to non-interacting fluids, since analyzing interactions would require the use of a much more involved and general formalism, developed by Carter in \cite{Carter1989} (and later expanded on by collaborators \cite{CarterLanglois1998,ComerLanglois1994})  for multiple fluids.

\subsection{Background and Radial Perturbations}

As we discussed in Sec.~\ref{ssec:back2}, all equations derived in that section can be easily generalized to arbitrarily many fluids simply by letting the index $X$ range over $1,2,...,k$, where $k$ is the number of independently conserved fluids.
The integration of the background equations can be generalized for multiple fluids in the same way as done for two fluids (see Sec.~\ref{ssec:back2}).
Specifically, we define an EoS $p_X(\epsilon_X)$ for each fluid, along with the value of the central energy density for each fluid $\epsilon_X^c$.
Equations~\eqref{eq:px}, \eqref{eq:2m} and \eqref{eq:2phi2} are integrated together with the multiple EoSs until some radius $R_1$ where $p_1(R_1)=0$.
At each inner radius $R_i$ ($i<k$, and we assume $R_i<R_j$ when $i<j$), we impose continuity of the remaining fluid functions and the gravitational field.
The gravitational potential $\Phi(r)$ at the outer radius $R_k$ is set by imposing continuity with the exterior solution $\Phi(R_k)=-\frac12\log(1-2m(R_k)/R_k)$.
The gravitational potential $\Lambda(r)$ is obtained from $m(r)$ (see Eq.~\eqref{eq:defm}), and the identity between the radial derivatives of $\Lambda$ and $\Phi$, given by Eq.~\eqref{eq:2diff} still holds.

The perturbation equations derived in Sec.~\ref{ssec:2pert} also easily generalize from two to multiple fluids.
The equation of motion for each fluid is once again obtained from the Euler equation, and it is given by Eq.~\eqref{eq:2eul}.
The expressions for $\partial_r\delta\Phi$, $\delta\Lambda$, $\delta p_X$ and $\delta\epsilon_X$ to be substituted into Eq.~\eqref{eq:2eul}
are once again given by Eqs.~\eqref{eq:2dphi}, \eqref{eq:2dlam}, \eqref{eq:2dp} and \eqref{eq:2de}, with $\bar\xi$ the density-averaged perturbation of Eq.~\eqref{eq:barxi}. The perturbation equations need to be solved within each radial region $r\in[R_i,R_{i+1}]$, with $\xi_i=0$ when $r>R_i$.

\subsection{Symmetric form}

Once we make all the required substitutions into the equations of motion,
we obtain a second-order differential equation for each fluid displacement vector $\xi_X$, along with a term that couples each fluid displacement vector to the other fluid displacement vectors.
As in the two-fluid case, we rewrite this system of equations in a symmetric (and in fact self-adjoint) form.
We accomplish this by first defining the variables $\phi_X=r^2e^{\Phi}\xi_X$ and multiplying the Euler equation by $e^{\Lambda-2\Phi}$, as was done for both the single-fluid and two-fluid cases.
See Appendix \ref{app:symm} for a more detailed argument.

We can then write the equation of motion for multiple fluids in terms of operators as done in Eq.~\eqref{eq:op},
\begin{equation}\label{eq:multiop}
    \hat A_X\ddot\phi_X+\hat C_X\phi_X+\sum_Y\hat D_{XY}\phi_Y=0 \ ,
\end{equation}
where the $\hat{D}_{XY}$ term represents coupling operators defined between each pair of fluids.
The operators are given by expressions similar to those for two fluids, 
\begin{subequations}\label{eq:multi_funcs}
\begin{gather}
    \hat A_X\phi_X\equiv A_X(r)\phi_X \ , \label{eq:multim}\\
    \hat C_X\phi_X\equiv -\partial_r\left(P_X(r)\partial_r\phi_X\right)+Q_X(r)\phi_X \ , \label{eq:multisl}\\
    \hat D_{XY}\phi_Y\equiv S_{XY}(r)\phi_Y+\partial_r(f_{XY}(r)\phi_Y)-f_{YX}(r)\partial_r\phi_Y \label{eq:multidxy} \ ,
\end{gather}
\end{subequations}
with the important relation that $S_{XY}(r)=S_{YX}(r)$.
For repeated indices, the operator $\hat D_{XX}=0$, and the functions associated with it, $S_{XX}(r)=0$ and $f_{XX}(r)=0$, are all zero.
The exact expressions for the functions $[A_X(r),P_X(r),Q_X(r),S_{XY}(r),f_{XY}(r)]$ are given in Appendix \ref{app:symm}.
Specifically, $A_X(r)$, $P_X(r)$, and $f_{XY}(r)$ have the same expression as in the two-fluid case [Eqs.~\eqref{eq:2AX}, \eqref{eq:2PX} and~\eqref{eq:fx} respectively], while $Q_X$ and $S_{XY}$ are given by Eqs.~\eqref{eq:Q_multi} and \eqref{eq:S_multi} respectively.
These functions have the same important properties as those which appear in the two-fluid case.
In particular, $[P_X(r),A_X(r),f_{XY}(r)]$ are all positive definite and $[P_X(R_X)=0,A_X(R_X)=0,Q_X(R_X)=0]$, 
while, for the interaction terms, $[f_{XY}(R_X)=f_{XY}(R_Y)=0,S_{XY}(R_X)=S_{XY}(R_Y)=0]$.
We emphasize that while the individual operators $\hat D_{XY}$ are not symmetric, the total \emph{system} [Eq.~\eqref{eq:multiop}] can be viewed as symmetric (as we will show in Sec.~\ref{ssec:multien}).

As in the case of one or two fluids, we search for solutions in terms of a mode expansion $\phi_X(t,r)=e^{i\omega t}\zeta_X(r)$.
With this, Eq.~\eqref{eq:multiop} reduces to
\begin{equation}\label{eq:slmulti}
    \omega^2\hat A_X\zeta_X=\hat C_X\zeta_X+\sum_Y\hat D_{XY}\zeta_Y \ .
\end{equation}
At the origin we impose the regularity conditions
\begin{equation}
    \lim_{r\to 0}\frac{\zeta_X}{r^2}= {\rm const} \ ,
\end{equation}
while at each radius $R_X$ we impose that the Lagrangian displacement of the pressure $p_X$ is zero, which we explicitly write as
\begin{equation}
\label{eq:lagrangian-displacement-multi-fluid}
\begin{split}
    \Delta p_X(R_X)&=
    \bigg[ -\frac{\gamma_Xp_X}{r^2}\frac{d\zeta_X}{dr}  \\
    & +4\pi \gamma_Xp_X\sum_Y(\epsilon_Y+p_Y)\frac{e^{2\Lambda-\Phi}}{r}(\zeta_X-\zeta_Y)\bigg]_{r=R_X}   \\
    &=0 \    .
\end{split}
\end{equation}
The condition in Eq.~\eqref{eq:lagrangian-displacement-multi-fluid}, along with the regularity of all quantities at each fluid boundary, imply that $\zeta_X$ and $\frac{d}{dr}\zeta_X$ are finite at $R_X$.
These conditions must apply for those fluids with $Y>X$, for which the pressure has not yet vanished.

\subsection{Energy, Lagrangian and Modes}\label{ssec:multien}

We next define the energy, Lagrangian, and their relation to the mode solutions for a generic multi-fluids star.
First, we introduce the unweighted $L^2$ inner product on the whole star
\begin{equation}
    \label{eq:def-inner-product-multi-fluid}
    \langle f,g\rangle=\int_0^{R_k}dr\, \bar f \; g \ .
\end{equation}
From the expressions for the operators $\hat A_X$ [Eq.~\eqref{eq:multim}] and $\hat C_X$ [Eq.~\eqref{eq:multisl}], we conclude that they are both symmetric (and, in fact, self-adjoint) in the $L^2$ space defined by Eq.~\eqref{eq:def-inner-product-multi-fluid}.
For the operator $\hat D_{XY}$, we have the following relation
\begin{equation}
\begin{split}
    &\langle \hat D_{XY}\psi_Y,\phi_X\rangle-\langle\psi_Y,\hat D_{YX}\phi_X\rangle\\
    =&\int_0^{R_k} dr\, \left[S_{XY}\bar\psi_Y\phi_X+\partial_r(f_{XY}\bar\psi_Y)\phi_X-f_{YX}\bar\psi_Y\partial_r\phi_X\right]\\
    -&\int_0^{R_k}dr\,\left[\bar\psi_YS_{YX}\phi_X+\bar\psi_Y\partial_r(f_{YX}\phi_X)-\bar\psi_Yf_{XY}\partial_r\phi_X\right] \ ,
\end{split}
\end{equation}
which is a generalization of Eq.~\eqref{eq:difd} for the two-component fluid star.
Since $S_{XY}=S_{YX}$ the multiplicative terms cancel, and from integration by parts, the derivative terms also cancel as it happened for two fluids (see Sec.~\ref{ssec:canon}).
Thus, we find that for any two fluids
\begin{equation}\label{eq:symcond}
    \langle \hat D_{XY}\psi_Y,\phi_X\rangle=\langle \psi_Y,\hat D_{YX}\phi_X\rangle \ ,
\end{equation}
just as in the two fluid case. 

We introduce symplectic structures for each fluid as done before in Eq.~\eqref{eq:symplec}.
The symplectic structure for the whole system is the sum 
\begin{equation}
    \label{eq:total-symplectic-structure-multifluid}
    W[\boldsymbol\psi,\boldsymbol\phi]=\sum_XW_X(\psi_X,\phi_X) \ .
\end{equation}
The total time derivative of Eq.~\eqref{eq:total-symplectic-structure-multifluid} is zero, which like in the two-fluid case, is a result of the properties of the operators $\hat D_{XY}$ [see Eq.~\eqref{eq:symcond}]. Mathematically, we can express this via
\begin{equation}
\begin{split}
    \frac{d}{dt}W\left[\boldsymbol{\psi},\boldsymbol{\phi}\right]
    &=
    \sum_{X,Y} \langle \hat D_{XY}\psi_Y,\phi_X\rangle-\langle \psi_Y,\hat D_{YX}\phi_X\rangle
    \nonumber\\
    &=
    0
    .
\end{split}
\end{equation}
The derivation of the above equation follows directly from the proof of Theorem \ref{th:symplectic}. 

The canonical energy is given by $E_c[\boldsymbol\phi,\dot{\boldsymbol\phi}]=-\frac12W[\boldsymbol\phi,\dot{\boldsymbol\phi}]$.
Once again, we can divide the canonical energy into a ``kinetic'' energy $T[\dot{\boldsymbol\phi}]$ and a ``potential'' energy $V[\boldsymbol\phi]$
\begin{subequations}
\begin{align}
    \label{eq:kinetic-energy-multifluid}
    T[\dot{\boldsymbol\phi}]&=\frac12\sum_X\langle \dot\phi_X,\hat A_X\dot\phi_X\rangle \ ,
    \\
    \label{eq:potential-energy-multifluid}
    V[\boldsymbol\phi]&=\frac12\sum_X\langle \phi_X,\hat C_X\phi_X\rangle+\frac12\sum_{X,Y}\langle \phi_X,\hat D_{XY}\phi_Y\rangle \ .
\end{align}
\end{subequations}
The Lagrangian is defined to be $\mathcal L[\boldsymbol\phi,\dot{\boldsymbol\phi}]=T[\dot{\boldsymbol\phi}]-V[\boldsymbol\phi]$, from which we can then define an action $S$, whose variation leads to the equations of motion [Eq.~\eqref{eq:multiop}].
As the kinetic and potential energies are equal to the sum of the single-fluid kinetic and potential energies, the derivation of the equations of motion follows just as for the two-fluid case (see Sec.~\ref{ssec:variation}).

The Rayleigh quotient for the multifluid system can be generalized similarly,
\begin{align}\label{eq:multi_rayleigh}
    I[\boldsymbol\zeta]
    &=
    \frac{V[\boldsymbol\zeta]}{T[\boldsymbol\zeta]}
    \nonumber\\
    &=
    \frac{\sum_{X}\langle \zeta_X,\hat C_X\zeta_X\rangle+\sum_{X,Y}\langle \zeta_X,\hat D_{XY}\zeta_Y\rangle}{\sum_X\langle\zeta_X,\hat A_X\zeta_X\rangle} \ .
\end{align}
The variation of Eq.~\eqref{eq:multi_rayleigh} results in a mode equation analogous to Eq.~\eqref{eq:slmulti}, but with $\omega^2=I[\Tilde{ \boldsymbol\zeta}]$, where $\Tilde{\boldsymbol\zeta}$ is a solution to the field equations.
In other words, Theorem \ref{th:vari} also applies to multi-fluid stars.
Since both $V[\boldsymbol\zeta]$ and $T[\boldsymbol\zeta]$ are real quantities, then the eigenvalues must be real.

To write Eq.~\eqref{eq:multi_rayleigh} in terms of a symmetric operator, we introduce a weighted Hilbert space $H_X$ for each fluid.
Each space $H_X$ has an $L^2$ inner product weighted by $A_X$ on the interval $(0,R_X)$.
Then, the Hilbert space we are interested in is
\begin{equation}
    H=\bigoplus_XH_X \ .
\end{equation}
As in Sec.~\ref{ssec:operator} we introduce matrix notation for weighting of the inner product 
\begin{equation}
    \mathbf A=\begin{pmatrix}\hat A_1&\cdots& 0\\\vdots&\ddots&\vdots\\0&\cdots&\hat A_k\end{pmatrix} \ ,
\end{equation}
that is, $\mathbf A_{ij}=\delta_{ij}\hat A_i$ in component notation (no summation), so that the inner product on $H$ can then be written as
\begin{equation}
    \braket{\boldsymbol\eta}{\boldsymbol\zeta}=\int dr\, \boldsymbol\eta^\dag \mathbf A\boldsymbol\zeta \ ,
\end{equation}
where the linear algebra operations are again understood above as performed with a flat metric. The bounds of the integral above are the same as in the case of two fluids. We can use the same matrix notation to define the potential energy operator,
\begin{equation}
    \mathcal V=\begin{pmatrix}\hat A_1^{-1}\hat C_1&\cdots &\hat A_1^{-1}\hat D_{1k}\\\vdots&\ddots&\vdots\\\hat A_k^{-1}\hat D_{k1}&\cdots&\hat A_k^{-1}\hat C_{k}\end{pmatrix} \ ,
\end{equation}
or equivalently, $\mathcal V_{ij}=\hat A_i^{-1}\hat D_{ij}$ for $i\ne j$ and $\mathcal V_{ii}=\hat A_i^{-1}\hat C_i$ for diagonal elements in component form.
The potential energy operator is symmetric due to Eq.~\eqref{eq:symcond}, and the Rayleigh quotient can then be written as
\begin{equation}
    I[\boldsymbol\zeta]=\frac{\braket{\boldsymbol\zeta}{\mathcal V\boldsymbol\zeta}}{\braket{\boldsymbol\zeta}} \ .
\end{equation}
Given this abstracted form, we see that, as in the two-fluid case, $\braket{\boldsymbol\zeta}{\mathcal V\boldsymbol\zeta}$ is a real quantity, and thus, the eigenvalues of $\mathcal V$ are real.
Similarly, one can show that the eigenvectors with different eigenvalues are orthogonal.

We will not prove here the results of Sec.~\ref{ssec:low} for the multi-fluid case,
but we expect the following to hold.
Due to the similarity with the two-fluid case,
we expect to be able to obtain a canonical self-adjoint extension.
Similarly, we expect that there is a lower bound to $I[\boldsymbol\zeta]$,
and that this bound can be derived similarly to what we did in Sec.~\ref{ssec:low}. Doing so would prove the existence of a fundamental mode and the self-adjoint Friedrichs extension for a multi-fluid case.

\section{Implications and Future Work\label{sec:conclusion}}

We have established the first set of rigorous conditions for the radial stability of two-fluid stars (summarized into Theorems \ref{th:symplectic}, \ref{th:vari} and \ref{th:fundamental}).
To achieve these results, we derived a canonical energy and a variational principle for generic perturbations of spherically symmetric stars (Eq.~\eqref{eq:2energy} and Eq.~\eqref{eq:2Rayleigh} respectively).
We then showed the existence of a fundamental mode for the perturbations of a two-fluid star, by providing a lower bound to its eigenfrequencies [Eq.~\eqref{eq:fund}].
While our results only hold for stars composed of fluids that interact with each other only gravitationally, such fluids can still (at least approximately) model systems of astrophysical interest, including dark-matter admixed stars~\cite{Tan2022}.
We also extended most of these results to fluids composed of an arbitrary number of fluids that interact with each other only gravitationally. 

Our work provides a formal justification for the stability criterion of two-fluid stars introduced in \cite{HENRIQUES1990} and used in other works, such as \cite{Hippert2023,Kain2021,GOLDMAN2013,DiGiovanni2022}.
The space of equilibrium configurations (i.e.~the configuration space) for two-fluid stars is described by the (positive) central energy density quarter plane, with each point of the plane corresponding to an equilibrium configuration, as shown in Fig.~\ref{fig:conf-space}. 
\begin{figure}[ht!]
    \includegraphics[clip=true,width=8.75cm]{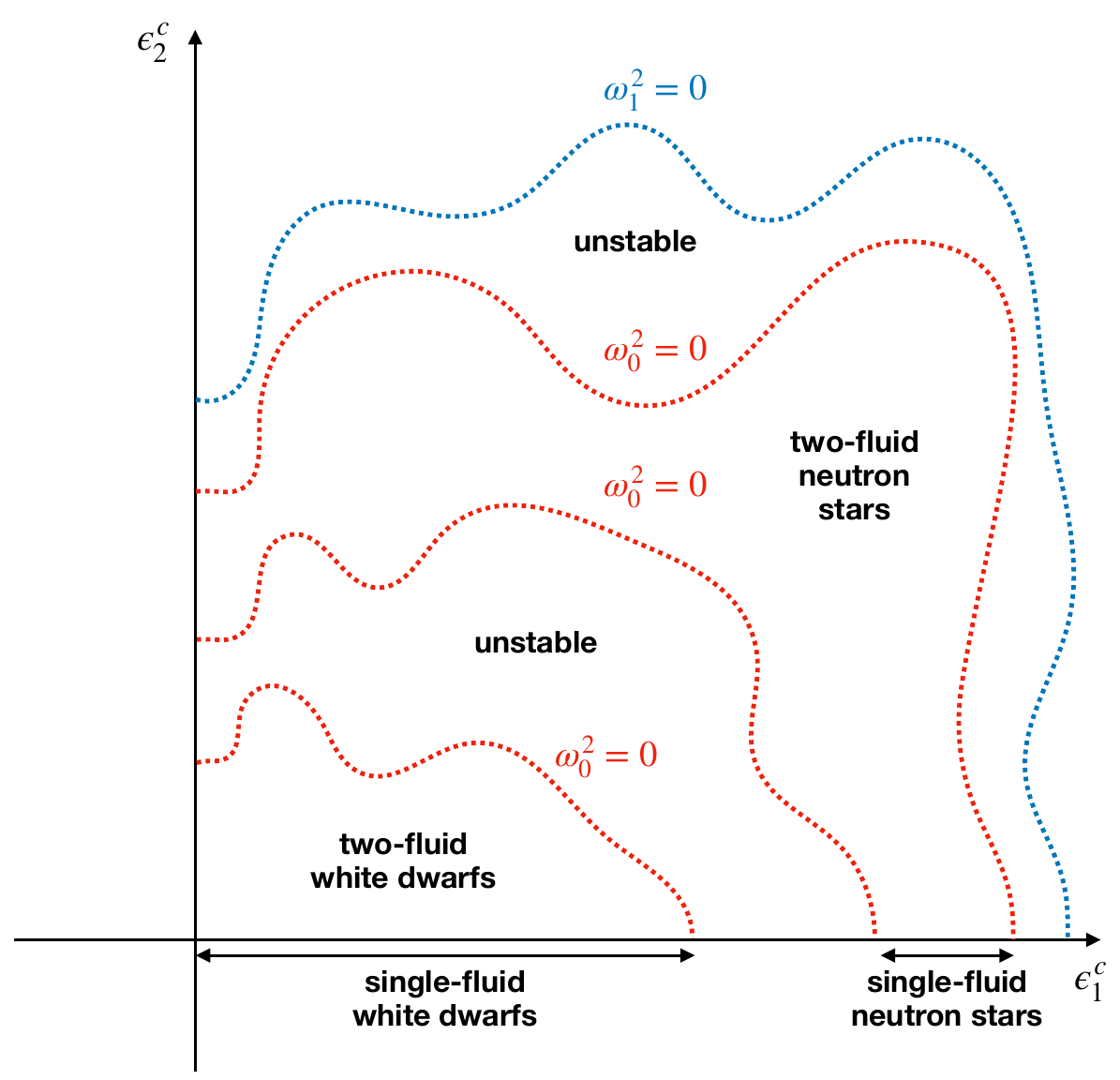}
    \caption{\label{fig:conf-space} Schematic representation of the configuration space. Each point in this plane represents an equilibrium  two-fluid star, characterized by certain observables. The red and blue dotted curves in this space represent extremal configurations for the fundamental and first-overtone mode, where the respective eigenfrequencies of the perturbations vanish. The areas enclosed by these curves define stable and unstable regions. }
\end{figure}
Each equilibrium configuration can be characterized by observable quantities $\{M,R_1,R_2,N_1,N_2\}$ (mass, radii, and total particle number for each fluid),
each of which can be thought of as a function of the central energy densities. Although these are equilibrium configurations, one does not know \textit{a priori} which are stable and which are unstable to perturbations. 

As we proved rigorously in this paper, the perturbations of each equilibrium configuration can be expressed as a sum of normal modes (because the perturbation are described by a self-adjoint problem), whose eigenfrequencies $\omega_n^2$ are bounded from below (Theorem \ref{th:fundamental}). 
Moreover, we also proved that for some mode labeled by its overtone number $n$, $\omega_n^2(\epsilon_1^c,\epsilon_2^c)$ corresponds to an extremal of the Rayleigh quotient, with $\omega_0^2(\epsilon_1^c,\epsilon_2^c)$ its minimum.
The condition $\omega_n^2(\epsilon_1^c,\epsilon_2^c)=0$ defines the $(\epsilon_1^c,\epsilon_2^c)$ curve that separates stable from unstable configurations for the $n$th mode,
i.e.~a configuration with higher central densities would force that mode to switch its stability properties and become unstable if previously stable or vice versa.  
Therefore, for a given $n$, the $\omega_n^2(\epsilon_1^c,\epsilon_2^c)=0$ condition defines a submanifold of co-dimension 1 in the configuration space, i.e.~a set of curves in the $(\epsilon_1^c,\epsilon_2^c)$ space, where the $n$th mode is extremal.
A graphical representation of these curves is shown in Fig.~\ref{fig:conf-space} for the fundamental mode (in red) and the first overtone (in blue).  
The union of all such submanifolds defines another submanifold of co-dimension 1,
i.e.~the union of all the extremal curves in the $(\epsilon_1^c,\epsilon_2^c)$ space, where there is at least one mode that is critical on each curve;
in Fig.~\ref{fig:conf-space}, this submanifold would be the union of all the red and blue dotted curves. 

Figure \ref{fig:conf-space} also shows how to use the results above to determine stability of astrophysical objects. White dwarfs reside at low enough central densities (in the lower left corner of Fig.~\ref{fig:conf-space}),
either when treated as a single-fluid star (i.e.~on the x-axis of the figure) or as a two-fluid star.
Since we know that white dwarfs are stable, then the area enclosed by the first, dotted red line and the $x$ and $y$ axis, corresponds to stable stellar configurations,
i.e.~stars whose perturbations do not grow exponentially because $\omega_0^2 > 0$ and $\omega_{n>0}^2 > \omega_0^2$.
As we consider higher central densities and cross this first boundary, we enter a regime where stellar configurations are unstable.
At even higher central densities, we cross a second, dotted red line that separates unstable from stable configurations.
Therefore, the area between the second and the third dotted red lines contains all two-fluid neutron star configurations that are stable, including the single-fluid sequence on the x-axis.
At even higher central densities, we again return to a regime of unstable stellar configurations.
The boundary between the two-fluid neutron star stable region and the higher central density unstable region is usually the separatrix that is searched for when studying stellar configurations with the highest masses and compactnesses. 

How can we relate this stability criteria to a simpler method to determine stability from the equilibrium configurations themselves? 
References~\cite{HENRIQUES1990,GOLDMAN2013, Hippert2023} argued that there is an $n$ mode for which $\omega_n^2(\epsilon_1^c,\epsilon_2^c)=0$ if and only if $\det(\partial N_i/\partial \epsilon_j^c)=0$.
The determinant condition is easier to employ because it relies only on information about the equilibrium (background) configurations.
For a single fluid, this determinant condition is equivalent to searching for extrema in the mass-central density curve, or equivalently, in the mass-radius curve.
Therefore, the determinant criteria is a generalization of the mass-radius criterion for single-fluid stars.
This generalization by itself, however, is not enough because it could have been possible that
(i) $\omega_n^2$ is not bounded from below (in which case there would not have been a fundamental mode),
or (ii) $\omega_n^2$ could have been complex (in which case the perturbative solutions would have been quasi-normal, with exponentially decaying or growing parts),
or (iii) the mode expansion of the perturbations could not have been complete (in which case there would have been other possible solutions to the perturbation equations that cannot be described by normal modes and could be unstable).
Our work here rigorously proves that none of these three possibilities are realized (for perfect fluids without phase transitions), thus establishing a robust connection between stability and the determinant condition. 

Although the determinant criterion (which is the generalization of the mass-radius criterion for single-fluid stars) is easy to use, the mass-radius criterion for single fluids fails in certain cases, and we expect the same to be true when considering more than one fluid. 
In particular, research has shown that first-order phase transitions can make the mass-radius maximum not coincide with the onset of radial instability \cite{DiClemente2020,Vasquez-Flores2012,Pereira2018}.
Moreover, first-order phase transitions can also lead to the existence of additional stability branches \cite{Gerlach1968,Alford2017,Tan2022}, in which case there will exist stable equilibrium configurations with higher central densities than the one usually found via the method of \cite{HENRIQUES1990}.
On the other hand, the calculation of eigenvalues will \textit{always} determine linear radial stability, provided the modes are complete.
Therefore, computing the fundamental mode and determining whether its square is positive or negative is the preferred method for determining stability in systems in which the determinant criterion fails.
However, one must be careful when computing the eigenvalues in such cases because the eigenvalue problem is altered when phase transitions are present, and appropriate boundary conditions must be employed across the phase transition \cite{Rau2023,Pereira2018,DiClemente2023}.



Our work, however, does not exhaust certain behaviors of the extremal curves. In particular, we have not proven that these curves are simple, closed and non-intersecting, which are necessary conditions for the curves to define ``interior'' and ``exterior'' regions.
We expect this to be necessary to extend, in a sensible manner, the mode-change criterion for modes beyond the fundamental mode of two-fluid stars (see \cite{BardeenThorneMeltzer1966} for the single-fluid case).
Particularly troublesome could be the results of \cite{Leung2012} in which certain eigenvalues found for a two-fluid star disappear when taking the single-fluid limit (and thus the curves are not closed).
All of this can be investigated through a deeper analysis of the spectrum of the operator $\mathcal V$ of Eq.~\eqref{eq:potoperator},
and, in particular, through the study of how this operator changes in the configuration space.

Our work also does not investigate the possibility that the spectrum of eigenvalues may be continuous, instead of discrete. The property of completeness of modes for a self-adjoint operator applies to the \textit{whole} spectrum, including any continuous part that may exist. If the spectrum were continuous, then the curves of Fig.~\ref{fig:conf-space} would become two-dimensional areas, and the analysis of stability would be much more involved. 
However, just as in the case for single-fluid stars, we expect there to be no continuous spectrum for radial modes \cite{DysonSchutz1979}. 

We have also not proven that one can extend Theorem \ref{th:fundamental} for radial eigenmodes of stars with three or more fluids or to two-fluid stars with phase transitions. 
In particular, it would be useful to relax the assumptions that $c_{s,X}^2\ne 0$ and that $c_{s,X}^2$ is differentiable within the inner radius of the star.
Stars with first-order phase transitions can violate these two assumptions \cite{Tan2022,Shahrbaf2020}, giving rise to secondary and higher stability branches. 
One could achieve such an extension by finding a better way to bound the $a_4$ term  of Eq.~\eqref{eq:a4}. 
Carrying out such an extension would therefore make the stability analysis of two-fluid stars with first-order phase transition as rigorous as that of stars without them. 

A final extension of this work that we would like to discuss involves the inclusion of non-gravitational interactions between the two fluids.
Such an extension would permit the study of a wider range of physical scenarios beyond dark-matter-admixed neutron stars, such as superfluid neutron stars \cite{AnderssonComer2001,AnderssonComerGrosart2004,CarterLanglois1998}.
This would involve using the formalism developed in \cite{Carter1989}, which leads to the equations in \cite{Comer1999}.
For the case with no ``entrainment'' ($\mathcal A=0$ in the equations in \cite{Comer1999}),
it would suffice to show that any additional terms stemming from the interactions do not affect the symmetric property of the operators (i.e.~Eq.~\eqref{eq:sym} still holds),
to obtain a canonical energy and a corresponding variational principle.
Additionally, one should show that any additional terms are also bounded from below, in the sense described in Sec.~\ref{ssec:bound}, to prove the existence of a fundamental mode.
The work presented in this paper provides a roadmap for such future extensions. 

\acknowledgments
We thank Jacquelyn Noronha-Hostler for helpful comments on this paper.
D. A. C. is thankful for the support from UIUC Graduate College and the Grainger College of Engineering, and from the Sloan Foundation.
J.~R.~and N.~Y.~also acknowledge support from the Simons Foundation through Award No. 896696, the National Science Foundation (NSF) Grant No. PHY-2207650 and NASA through Grant No. 80NSSC22K0806.



\appendix

\section{Writing into symmetric form for multiple fluids}\label{app:symm}

In this appendix, we provide more details on how to derive the mode equations for the multi-fluid description, i.e.~Eqs.~\eqref{eq:multim}, \eqref{eq:multisl} and \eqref{eq:multidxy}, which appear in Eq.~\eqref{eq:multiop},
To obtain the two-fluid description, it suffices to take the two-fluid limit of the results of this appendix.
We also provide here explicit forms of the functions $[A_X(r),P_X(r),Q_X(r),S_{XY}(r),f_{XY}(r)]$ for multiple fluids. As we mentioned in Sec.~\ref{ssec:2pert}, these expressions can be found by substituting Eqs.~\eqref{eq:2dphi}, \eqref{eq:2dlam}, \eqref{eq:2dp} and \eqref{eq:2de} into Eq.~\eqref{eq:2eul},
using the variable $\phi_X=r^2e^{\Phi}\xi_X$ and multiplying both sides of the Eq.~\eqref{eq:2eul} by $e^{\Lambda-2\Phi}$.
To simplify the presentation of our calculations, we decompose our expressions into modes; i.e. $\phi_X(t,r)\to e^{i\omega t}\zeta_X(r)$.
Then, the full expression for Eq.~\eqref{eq:multiop} is
\begin{widetext}
\begin{equation}\label{eq:multi2}
\begin{split}
    -\omega^2\frac{\epsilon_X+p_X}{r^2}e^{3\Lambda-\Phi}\zeta_X=&-e^{\Lambda-2\Phi}\frac{d}{dr}\left[-\frac{(\epsilon_X+p_X)\Phi'}{r^2}e^{-\Phi}\zeta_X-\frac{\gamma_X p_Xe^{-\Phi}}{r^2}\frac{d\zeta_X}{dr}+4\pi\gamma_Xp_X\frac{e^{2\Lambda-\Phi}}{r}\sum_{Y}(\epsilon_Y+p_Y)(\zeta_Y-\zeta_X)\right]\\
    &+\Phi'\frac{\gamma_Xp_X+(\epsilon_X+p_X)}{\gamma_Xp_X}\left[-\frac{(\epsilon_X+p_X)\Phi'}{r^2}e^{\Lambda-3\Phi}\zeta_X-\frac{\gamma_Xp_X}{r^2}e^{\Lambda-3\Phi}\frac{d\zeta_X}{dr}\right]\\
    &+\Phi'[\gamma_Xp_X+(\epsilon_X+p_X)]\frac{4\pi}{r}e^{3\Lambda-3\Phi}\sum_Y(\epsilon_Y+p_Y)(\zeta_Y-\zeta_X)\\
    &-(\epsilon_X+p_X)\frac{4\pi}{r}e^{3\Lambda-3\Phi}\left[\sum_Y(\epsilon_Y+p_Y)\left(\Phi'-\frac1r\right)\zeta_Y-\gamma_Yp_Y\frac{d\zeta_Y}{dr}\right]\\
    &-(\epsilon_X+p_X)\frac{4\pi}{r}e^{3\Lambda-3\Phi}\sum_{Y,Z}\gamma_Yp_Y4\pi re^{2\Lambda}(\epsilon_Z+p_Z)(\zeta_Z-\zeta_Y) \ .
\end{split}
\end{equation}
\end{widetext}
Our goal is to show that Eq.~\eqref{eq:multi2} can be written as Eq.~\eqref{eq:multiop}.
We see that for the two equations to match, the term multiplying the eigenvalue must be identified with $A_X$, that is
\begin{equation}\label{eq:A_app}
    A_X(r)=\frac{\epsilon_X+p_X}{r^2}e^{3\Lambda-\Phi} \ .
\end{equation}

First, we look to obtain a term of the form $\frac{d}{dr}(P_X\frac{d}{dr}\zeta_X)$ in Eq.~\eqref{eq:multisl}.
We can bring the $e^{\Lambda-2\Phi}$ into the parenthesis of the first line of Eq.~\eqref{eq:multi2} by using the product rule 
\begin{equation}\label{eq:productrule}
\begin{split}
    e^{\Lambda-2\Phi}\frac{d}{dr}\left(\frac{\gamma_Xp_Xe^{-\Phi}}{r^2}\frac{d\zeta_X}{dr}\right)=\frac{d}{dr}\left(\frac{\gamma_Xp_X}{r^2}e^{\Lambda-3\Phi}\frac{d\zeta_X}{dr}\right)\\
    -\frac{\gamma_Xp_Xe^{\Lambda-3\Phi}}{r^2}\frac{d\zeta_X}{dr}\left[\frac{d}{dr}(\Lambda-\Phi)-\Phi'\right] \ .
\end{split}
\end{equation}
We now compare with Eq.~\eqref{eq:multi_funcs} from which we can identify that
\begin{equation}\label{eq:P_app}
    P_X(r)=\frac{\gamma_Xp_X}{r^2}e^{\Lambda-3\Phi} \ .
\end{equation}
Note that both $A_X$ and $P_X$ have the same form as in the single-fluid case [see Eq.~\eqref{eq:funcs1}].
For the term in the second line of Eq.~\eqref{eq:productrule}, we use Eq.~\eqref{eq:2diff} to obtain
\begin{equation}
    \frac{\gamma_Xp_X}{r^2}e^{\Lambda-3\Phi}\frac{d}{dr}(\Lambda-\Phi)=4\pi \frac{\gamma_Xp_X}{r}e^{3\Lambda-3\Phi}\sum_Y(\epsilon_Y+p_Y) \ .
\end{equation}
Substituting this result into Eq.~\eqref{eq:multi2}, we obtain that all remaining first derivative terms not included in $\frac{d}{dr}(P_X\frac{d}{dr}\zeta_X)$ vanish.

Therefore, all the terms in the right-hand side of Eq.~\eqref{eq:multi2} that act on $\zeta_X$ can be written as a Sturm-Liouville operator
\begin{equation}
    \label{eq:inter-multi-eom-appendix}
    \hat C_X\zeta_X=-\frac{d}{dr}\left(P_X(r)\frac{d}{dr}\zeta_X\right)+Q_X(r)\zeta_X \ ,
\end{equation}
with $Q_X$ given by all multiplicative terms of $\zeta_X$ in Eq.~\eqref{eq:multi2}.
Using the product rule similarly as above, substituting in the background equations to simplify, and expanding and cancelling terms, we find that the coefficient of $\zeta_X$ in Eq.~\eqref{eq:inter-multi-eom-appendix} can be written as
\begin{widetext}
\begin{equation}\label{eq:Q_multi}
\begin{split}
        Q_X(r)=&-e^{\Lambda-3\Phi}\frac{(\epsilon_X+p_X)}{r^2}\left[\Phi''-\frac2r\Phi'-(\Phi')^2\right]+4\pi e^{3\Lambda-3\Phi}\frac{(\epsilon_X+p_X)^2}{r}\left(\Phi'-\frac1r\right)\\
        &+\sum_{Y\ne X}4\pi e^{3\Lambda-3\Phi}\frac{(\epsilon_X+p_X)(\epsilon_Y+p_Y)}{r}-\frac{d}{dr}\left[4\pi e^{3\Lambda-3\Phi}\frac{\gamma_Xp_X(\epsilon_Y+p_Y)}{r}\right]\\
        &+16\pi^2 e^{5\Lambda-3\Phi}\sum_{Y\ne X}\left[(\epsilon_X+p_X)^2\gamma_Yp_Y-\gamma_Xp_X(\epsilon_X+p_X)(\epsilon_Y+p_Y)+\gamma_Xp_X(\epsilon_Y+p_Y)\left(\sum_Z(\epsilon_Z+p_Z)\right)\right] \ .
\end{split}
\end{equation}
\end{widetext}
The first line of Eq.~\eqref{eq:Q_multi} is the $Q_X$ coefficient for a single fluid (see Eq.~\eqref{eq:Qb}). 
Lines two and three are due to the gravitational interaction between each fluid. 
In the specific case of two fluids, the last line results in only two terms, which together are symmetric under index interchange.
Likewise, the first term in the second line is also symmetric under index interchange,
which allowed for writing $Q_X$ as the sum of two terms, one that depends on the index and one which does not (called $R$), as given in Eqs.~\eqref{eq:2Qb}-\eqref{eq:Rb}.

The $\zeta_Z$ term in Eq.~\eqref{eq:multi2} can just be changed to a $\zeta_Y$ since $Y$ and $Z$ are just dummy indices under the sum,
\begin{widetext}
\begin{equation}
\begin{split}
    &-16\pi^2(\epsilon_X+p_X)e^{5\Lambda-3\Phi}\sum_{Y\ne X}\gamma_Yp_Y\sum_{Z\ne X,Y}(\epsilon_Z+p_Z)\zeta_Z
    =-16\pi^2(\epsilon_X+p_X)e^{5\Lambda-3\Phi}\sum_{Y\ne X}(\epsilon_Y+p_Y)\zeta_Y\left(\sum_{Z\ne Y,X}\gamma_Zp_Z\right) \ .
\end{split}
\end{equation}
\end{widetext}
Thus, all the remaining terms that need to be taken into account in Eq.~\eqref{eq:multi2} can be written as some operator acting on all the other fluid variables.
Mathematically, this is expressed using operators $\hat D_{XY}$ and summing $\sum_{Y\ne X}\hat D_{XY}\zeta_Y$ for the operator $\hat D_{XY}$ given by
\begin{widetext}
\begin{equation}\label{eq:fulldxy}
\begin{split}
    \hat D_{XY}\zeta_Y=+&e^{\Lambda-2\Phi}\frac{d}{dr}\left[4\pi \gamma_Xp_X\frac{e^{2\Lambda-\Phi}}{r}(\epsilon_Y+p_Y)\zeta_Y\right]-\frac{4\pi}{r}e^{3\Lambda-3\Phi}(\epsilon_X+p_X)\gamma_Yp_Y\frac{d\zeta_Y}{dr}\\
    +&\frac{4\pi}{r}e^{3\Lambda-3\Phi}(\epsilon_Y+p_Y)\left\{-\Phi'\left[\gamma_Xp_X+(\epsilon_X+p_X)\right]+(\epsilon_X+p_X)\left(\Phi'-\frac1r\right)\right\}\zeta_Y\\
    +&16\pi^2e^{5\Lambda-3\Phi}(\epsilon_X+p_X)\left[(\epsilon_Y+p_Y)\left(\sum_{Z\ne X,Y}\gamma_Zp_Z\right)-\gamma_Yp_Y\left(\sum_{Z\ne X,Y}(\epsilon_Z+p_Z)\right)\right]\zeta_Y \ .
\end{split}
\end{equation}
\end{widetext}

Just as we did for the operator $\hat C_X$, we look to express the above equation in a simpler form.
Using the product rule, we can move the $e^{\Lambda-2\Phi}$ inside of the derivative in the first line (akin to Eq.~\eqref{eq:productrule}).
Then, we can define the function $f_{XY}$ as the term inside the derivative of the first line of Eq.~\eqref{eq:fulldxy},
\begin{equation}\label{eq:fxy}
    f_{XY}=4\pi e^{3\Lambda-3\Phi}\frac{\gamma_Xp_X(\epsilon_Y+p_Y)}{r} \ ,
\end{equation}
which has the same form as in the two-fluid case, which we called $f_X$ (see Eq.~\eqref{eq:fx}).
The only remaining term with a derivative on $\zeta_Y$ of Eq.~\eqref{eq:fulldxy} can incidentally be written as $f_{YX}\frac{d}{dr}\zeta_Y$ (second term of the first line).
All remaining terms in $D_{XY}\zeta_Y$ must then be multiplicative terms, and thus, we write
\begin{equation}\label{eq:inter-fluid-d-appendix}
    \hat D_{XY}\zeta_Y=\frac{d}{dr}(f_{XY}\zeta_Y)-f_{YX}\frac{d}{dr}\zeta_Y+S_{XY}\zeta_Y \ .
\end{equation}

The $S_{XY}$ term must include (aside from the explicit multiplicative terms in Eq.~\eqref{eq:fulldxy}) the term that shows up from the product rule
$-f_{XY}\zeta_Y\left[\frac{d}{dr}(\Lambda-\Phi)-\Phi'\right]$. We must then substitute for $\frac{d}{dr}\Lambda-\frac{d}{dr}\Phi$ via Eq.~\eqref{eq:2diff}, with the sum indexed by $Z$.
After cancellations, we find that the multiplicative terms are
\begin{widetext}
\begin{equation}\label{eq:S_multi}
\begin{split}
    S_{XY}=&-4\pi e^{3\Lambda-3\Phi}\frac{(\epsilon_X+p_X)(\epsilon_Y+p_Y)}{r^2}-16\pi^2e^{5\Lambda-3\Phi}\left[\gamma_Xp_X(\epsilon_Y+p_Y)^2+\gamma_Yp_Y(\epsilon_X+p_X)^2\right]\\
    &-16\pi^2e^{5\Lambda-3\Phi}\left[\gamma_Xp_X(\epsilon_Y+p_Y)+\gamma_Yp_Y(\epsilon_X+p_X)\right]\left[\sum_{Z\ne X,Y}(\epsilon_Z+p_Z)\right] \\
    &+16\pi^2e^{5\Lambda-3\Phi}(\epsilon_X+p_X)(\epsilon_Y+p_Y)\left(\sum_{Z\ne X,Y}\gamma_Zp_Z\right) \ ,
\end{split}
\end{equation}
\end{widetext}
Note that the function $S_{XY}$ is symmetric under index interchange -- i.e. $S_{XY}=S_{YX}$.
Moreover, for the specific case of two fluids, all the sums over $Z\ne X,Y$ are empty, so only the first line contributes, giving us exactly Eq.~\eqref{eq:Sb}.

To conclude, we can write Eq.~\eqref{eq:multi2} in the form of Eq.~\eqref{eq:multisl} as follows:
$\hat A_X\zeta_X=A_X(r)\zeta_X$ and $A_X$ is given by Eq.~\eqref{eq:A_app}; $\hat C_X\zeta_X$ is  given by Eq.~\eqref{eq:inter-multi-eom-appendix} with the functions $P_X,Q_X$ given by Eqs.~\eqref{eq:P_app} and \eqref{eq:Q_multi} respectively; and $\hat D_{XY}$ is given by Eq.~\eqref{eq:inter-fluid-d-appendix} with the functions $f_{XY},S_{XY}$ as given in Eqs.~\eqref{eq:fxy} and \eqref{eq:S_multi} respectively.


\section{Bound for $a_4$ term}\label{app:lastbound}

Here, we show that the $a_4$ term in Eq.~\eqref{eq:a_4} is bounded from below by the expression in Eq.~\eqref{eq:a4}.
First, let us expand $a_4$ and separate it into three terms ($a_4^1,a_4^2,a_4^0$), where
\begin{align}
    a_4^{1,2}&=-\int dr\, \frac{d}{dr}f_{1,2}(\zeta_{1,2})^2 \ ,\\
     a_4^0&=-2\int dr\,f_1\zeta_1'\zeta_2+f_2\zeta_1\zeta_2' \ , \label{eq:appa4}
\end{align}
where $\zeta_X' \equiv \frac{d}{dr}\zeta_X$ and $f_X'\equiv \frac{d}{dr}f_X$.

To obtain a bound for $a_4^{1,2}$, we simply take the derivative of $f_X$ according to Eq.~\eqref{eq:fx}.
When doing this, we need to take the derivative of $\gamma_Xp_X$, which (by virtue of Eq.~\eqref{eq:relc}) is
\begin{equation}
\begin{split}
    \frac{d}{dr}(\gamma_Xp_X)=&\frac{d}{dr}[c_{s,X}^2(\epsilon_X+p_X)]\\
    =&\Phi'(\epsilon_X+p_X)\left[1+\frac{(\epsilon_X+p_X)}{\gamma_Xp_X}\right]c_{s,X}^2\\
    &+(\epsilon_X+p_X)\left(\frac{d}{dr}c_{s,X}^2\right) \ .
\end{split}
\end{equation}
Then, using Eq.~\eqref{eq:2diff}, we find that the derivative of $f_X$ is
\begin{widetext}
\begin{equation}\label{eq:dfx}
\begin{split}
    \frac{d}{dr}f_X=&-\frac{4\pi}{r} e^{3\Lambda-3\Phi}(\epsilon_X+p_X)(\epsilon_Y+p_Y)c_{s,X}^2\left[\frac1r-\Phi'\left(2+\frac{\epsilon_X+p_X}{\gamma_Xp_X}+\frac{\epsilon_Y+p_Y}{\gamma_Yp_Y}\right)\right]\\
    &+48\pi^2 e^{5\Lambda-3\Phi}(\epsilon_X+p_X)(\epsilon_Y+p_Y)c_{s,X}^2\left[\sum_Z(\epsilon_Z+p_Z)\right]+4\pi e^{3\Lambda-3\Phi}\frac{(\epsilon_X+p_X)(\epsilon_Y+p_Y)}{r}\left(\frac{d}{dr}c_{s,X}^2\right) \ ,
\end{split}
\end{equation}
\end{widetext}
where we separated the result into two lines, with the first line being negative definite,
while, in the second line, the first term is positive definite and the second term can be either negative or positive.
Since the first terms are negative definite, their resulting terms in $a_4^{1,2}$ will be positive definite, and thus, bounded by zero.
Hence, we have that $a_4^1$ is bounded by
\begin{equation}
\label{eq:sup-inter-a41}
\begin{split}
    a_4^1\geq-\int dr\, &e^{\Lambda-3\Phi}\frac{(\epsilon_1+p_1)}{r^2}(\zeta_1)^2\\
    &\times\Bigg[4\pi re^{2\Lambda}(\epsilon_2+p_2)\left(\frac{d}{dr}c_{s,1}^2\right)\\
    &+48\pi^2r^2e^{4\Lambda}(\epsilon_2+p_2)c_{s,1}^2\left(\sum_Z(\epsilon_Z+p_Z)\right)\Bigg] \ ,
\end{split}
\end{equation}
and a similar result holds for $a_4^2$.
Then, similar to the method in Sec.~\ref{ssec:low},
if the term inside the large square brackets has a supremum, Eq.~\eqref{eq:sup-inter-a41} gives a lower bound for $a_4^1$.
Using that $(\epsilon_X+p_X)\leq \rho_X^{\rm max},c_{s,X}^2\leq 1$ and $\mathcal I_X$ is given by Eq.~\eqref{eq:intex}, we obtain
\begin{equation}
\begin{split}
    a_4^1\geq &-4\pi \mathcal I_1\rho_2^{\rm max}\sup\left(re^{2\Lambda}\frac{d}{dr}c_{s,1}^2\right)\\
    &-48\pi^2\mathcal I_1\rho_2^{\rm max}(\rho_1^{\rm max}+\rho_2^{\rm max})\sup\left(r^2e^{4\Lambda}\right) \ ,
\end{split}
\end{equation}
and similarly for $a_4^2$.
Then, the sum of the two terms is bounded by
\begin{equation}
\begin{split}
    a_4^1+a_4^2\geq &-48\pi^2(\rho_1^{\rm max}+\rho_2^{\rm max})\sup(r^2e^{4\Lambda})(\mathcal I_1\rho_2^{\rm max}+\mathcal I_2\rho_1^{\rm max})\\
    &-4\pi \mathcal I_1\rho_2^{\rm max}\sup(re^{2\Lambda}\frac{d}{dr}c_{s,1}^2)\\
    &-4\pi \mathcal I_2\rho_1^{\rm max}\sup(re^{2\Lambda}\frac{d}{dr}c_{s,2}^2) \ .
\end{split}
\end{equation}
Since $\mathcal I_X$ and $\rho_X$ are both positive definite, in the case that both $\sup(re^{2\Lambda}\frac{d}{dr}c_{s,X}^2)$ are also positive,
the following inequality holds:
\begin{equation}\label{eq:semicomb}
\begin{split}
    &-\left[\mathcal I_1\rho_2^{\rm max}\sup(re^{2\Lambda}\frac{d}{dr}c_{s,1}^2)+\mathcal I_2\rho_1^{\rm max}\sup(re^{2\Lambda}\frac{d}{dr}c_{s,2}^2)\right]\\
    \geq &-(\mathcal I_1\rho_2^{\rm max}+\mathcal I_2\rho_1^{\rm max})\left[\sup(re^{2\Lambda}\frac{d}{dr}c_{s,1}^2)+\sup(re^{2\Lambda}\frac{d}{dr}c_{s,2}^2) \right] \ .
\end{split}
\end{equation}
Generally speaking, $\frac{d}{dr}c_{s,X}^2$ can be negative everywhere.
However, this does not pose a problem because in the case that $\sup(re^{2\Lambda}\frac{d}{dr}c_{s,X}^2)\leq 0$, the expression multiplying it will be bounded from below by zero.
We can express this choice by defining $\alpha_X$ as in Eq.~\eqref{eq:alphax}, so that we find
\begin{equation}\label{eq:a1a2}
\begin{split}
    a_4^1+a_4^2\geq &    
    -48\pi^2(\mathcal I_1\rho_2^{\rm max}+\mathcal I_2\rho_1^{\rm max})(\rho_1^{\rm max}+\rho_2^{\rm max})\sup(re^{4\Lambda})\\
    &-4\pi(\mathcal I_1\rho_2^{\rm max}+\mathcal I_2\rho_1^{\rm max})(\alpha_1+\alpha_2) \ .
\end{split}
\end{equation}

Let us now focus on the $a_4^0$ term and define $f_q$ when given a function $c_q^2(r)\in[0,1]$ by
\begin{equation}\label{eq:fq}
    f_q=4\pi re^{3\Lambda-3\Phi}\frac{(\epsilon_1+p_1)(\epsilon_2+p_2)}{r}c_q^2 \ .
\end{equation}
In particular, for $c_q^2=c_{s,X}^2$, we find that $f_q=f_X$.
Notice that, if we integrate by parts our definition for $a_4^0$ in Eq.~\eqref{eq:appa4} and if we apply the boundary conditions, we find that
\begin{equation}
    a_4^0=2\int dr\, f_1\zeta_1\zeta_2'+f_2\zeta_1'\zeta_2+2\int dr\, (f_1+f_2)'(\zeta_1\zeta_2) \ .
\end{equation}
The second integral can be integrated by parts, and defining the mean speed of sound $c_A^2=\frac12(c_{s,1}^2+c_{s,2}^2)$ and its respective $f_A$ (by Eq.~\eqref{eq:fq}), we then find that $a_4^0$ is also expressible as
\begin{equation}\label{eq:alta4}
    a_4^0=2\int dr\, f_1\zeta_1\zeta_2'+f_2\zeta_1'\zeta_2-4\int dr\, f_A(\zeta_1\zeta_2)'
\end{equation}
Now let us fix $c_q^2$, and suppose we have a vector $(\zeta_1,\zeta_2)$, such that
\begin{equation}\label{eq:cond}
    \int dr\, [(f_1-f_q)\zeta_1'\zeta_2+(f_2-f_q)\zeta_1\zeta_2']\geq 0 \ .
\end{equation}
Then, from the definition of $a_4^0$ in Eq.~\eqref{eq:appa4}, 
\begin{equation}\label{eq:a401}
    a_4^0\geq -2\int dr\, f_q(\zeta_1\zeta_2)' \ .
\end{equation}
Alternatively, suppose the vector satisfies the negation of Eq.~\eqref{eq:cond} (i.e. changing $\geq$ to $\leq$).
Then, if we substitute this into the alternative form of $a_4^0$ in Eq.~\eqref{eq:alta4}, we find
\begin{equation}\label{eq:a402}
    a_4^0\geq 2\int dr\, f_q(\zeta_1\zeta_2)'-4\int dr\, f_A(\zeta_1\zeta_2)' \ .
\end{equation}
In particular, if we choose, $f_q=f_A$, both Eqs.~\eqref{eq:a401} and \eqref{eq:a402} are the same.
Since either Eq.~\eqref{eq:cond} or its negation must be satisfied, we find that for any vector $(\zeta_1,\zeta_2)$, we must have
\begin{equation}\label{eq:a40}
    a_4^0\geq -2\int dr\, f_A(\zeta_1\zeta_2)'=\int dr\, (f_1'+f_2')(\zeta_1\zeta_2) \ .
\end{equation}

When finding a lower bound for $a_4^{1,2}$, we saw how to treat terms like $\zeta_1'$, while in Sec.~\ref{ssec:low} we saw how to treat terms that are a multiplication of $\zeta_1\zeta_2$.
Then, combining both approaches together, we can find a bound for $a_4^0$.
To accomplish this, let us look back at Eq.~\eqref{eq:dfx}, where we have terms that are either positive definite, negative definite or depend on $\frac{d}{dr}c_{s,X}^2$.
Then, we can write
\begin{equation}
    f_X'(r)=g_X(r)-h_X(r)+l_X(r) \ ,
\end{equation}
such that $g_X(r),h_X(r)$ are both positive definite and
\begin{equation}
    l_X(r)=4\pi e^{3\Lambda-3\Phi}\frac{(\epsilon_X+p_X)(\epsilon_Y+p_Y)}{r}\frac{d}{dr}c_{s,X}^2 \ .
\end{equation}
In the case that $\frac{d}{dr}c_{s,X}^2$ is negative definite (e.g. for a polytrope),
then the $l_X(r)$ term can just be included within $h_X(r)$ and the following calculations become easier.
Let $g=g_1+g_2,h=h_1+h_2,l=l_1+l_2$, so that we then obtain
\begin{equation}
    \begin{split}
        a_4^0\geq &\;\int dr\, (g-h+l)(\zeta_1\zeta_2)\\
        \geq& \frac12\int dr\, \left[g(\zeta_1+\zeta_2)^2+h(\zeta_1-\zeta_2)^2\right]\\
         &-\frac12\int dr\, (g+h)\left[(\zeta_1)^2+(\zeta_2)^2\right]+\int dr\, l\zeta_1\zeta_2 \ .
    \end{split}
\end{equation}
The integrand $g(\zeta_1+\zeta_2)^2+h(\zeta_1-\zeta_2)^2$ is positive definite, and hence we are left with only the second line to bound.
Let us explicitly write $(g,h)$ (from Eq.~\eqref{eq:dfx}) as
\begin{equation}
\begin{split}
    g(r)&=96\pi^2 e^{5\Lambda-3\Phi}(\epsilon_1+p_1)(\epsilon_2+p_2)c_A^2\left[\sum_Z(\epsilon_Z+p_Z)\right] \ ,\\
    h(r)&=\frac{8\pi}{r}e^{3\Lambda-3\Phi}(\epsilon_1+p_1)(\epsilon_2+p_2)c_A^2\left[\frac{1}{r}-2\Phi'\left(1+\frac{1}{c_H^2}\right)\right] \ , 
\end{split}
\end{equation}
where $c_H^2$ is the harmonic mean of the speeds of sound, i.e. $\frac{2}{c_H^2}=\frac{1}{c_{s,1}^2}+\frac{1}{c_{s,2}^2}$.
Carrying out the procedure of taking the supremum, we then find
\begin{equation}\label{eq:b22}
\begin{split}
    a_4^0\geq &-4\pi\left(\mathcal I_1\rho_2^{\rm max}+\mathcal I_2\rho_1^{\rm max}\right)
    \\ &\times \sup\left[c_A^2e^{2\Lambda}(1-2r\Phi')-2re^{2\Lambda}\Phi'\frac{c_A^2}{c_H^2}\right]\\
    &-48\pi^2\left(\mathcal I_1\rho_2^{\rm max}+\mathcal I_2\rho_1^{\rm max}\right)(\rho_1^{\rm max}+\rho_2^{\rm max})\sup(c_A^2r^2e^{4\Lambda})\\
    &+\int dr\, l\zeta_1\zeta_2 \ .
\end{split}
\end{equation}
Unfortunately, since $c_A^2\geq c_H^2$, we cannot eliminate the $\frac{c_A^2}{c_H^2}$ term using inequalities.
Because of this, we cannot put an upper bound on $\sup\left(\frac{c_{s,1}^2}{c_{s,2}^2}+\frac{c_{s,2}^2}{c_{s,1}^2}\right)$.
Thus, for this supremum to exist, we must require $c_{s,1}^2,c_{s,2}^2>0$ at least on $(0,R_1)$.

Finally, we are left with the $l$ term.
As mentioned before, this term is negative definite in the simplest models, and
hence we treat it similarly to $h(r)$ by writing
\begin{equation}\label{eq:ll}
    \int dr\, l\zeta_1\zeta_2=\frac12\int dr\, l(\zeta_1^2+\zeta_2^2)-\frac12\int dr\, l(\zeta_1-\zeta_2)^2 \ .
\end{equation}
We write out explicitly the first term and put a lower bound on it in the following way:
\begin{equation}\label{eq:beta}
    \begin{split}
        &4\pi \int dr\, re^{2\Lambda}(\epsilon_1+p_1)(\epsilon_2+p_2)\frac{e^{\Lambda-3\Phi}}{r^2}
        \\ &\quad\quad\times
        \left(\frac{d}{dr}c_{s,1}^2+\frac{d}{dr}c_{s,2}^2\right) [(\zeta_1)^2+(\zeta_2)^2]\\
        \geq &4\pi \inf\left[re^{2\Lambda}\left(\frac{d}{dr}c_{s,1}^2+\frac{d}{dr}c_{s,2}^2\right)\right]\left(\mathcal I_1\rho_2^{\rm max}+\mathcal I_2\rho_1^{\rm max}\right) \ .
    \end{split}
\end{equation}
To be able to write out the above as proportional to $\mathcal I_1+\mathcal I_2$ (which we later explain how to do through Eq.~\eqref{eq:combine}),
we want $\inf[re^{2\Lambda}(\frac{d}{dr}c_{s,1}^2+\frac{d}{dr}c_{s,2}^2)]$ to be negative definite.
However, in the case that it is positive, then it is bounded from below by zero.
This is similar to what happened in Eq.~\eqref{eq:semicomb}, where we introduced $\alpha_X$ as in Eq.~\eqref{eq:alphax}.
In the same way as before, we define $\beta_X$ as in Eq.~\eqref{eq:betax} and $\beta_A$ to be the same expression but with the mean of the speed of sounds $c_A^2$ instead of $c_{s,X}^2$.

Let us look next at the second term of Eq.~\eqref{eq:ll}.
Using supremums, we find that
\begin{equation}\label{eq:lastsquare}
\begin{split}
    &-4\pi\int dr\, re^{2\Lambda}(\epsilon_1+p_1)(\epsilon_2+p_2)\frac{e^{\Lambda-3\Phi}}{r^2}
    \\ &\quad\quad\times
    \left[\frac{d}{dr}c_{s,1}^2+\frac{d}{dr}c_{s,2}^2\right](\zeta_1-\zeta_2)^2\\
    \geq&-8\pi\alpha_A\int dr\, \frac{e^{\Lambda-3\Phi}}{r^2}(\epsilon_1+p_1)(\epsilon_2+p_2)(\zeta_1-\zeta_2)^2 \ ,
\end{split}
\end{equation}
where $\alpha_A$ is given by Eq.~\eqref{eq:alphax} but with $c_A^2$ replacing $c_{s,X}^2$.
In the case that $\frac{d}{dr}c_{s,1}^2+\frac{d}{dr}c_{s,2}^2$ is negative definite, the expression is bounded by zero.
Otherwise, we have to expand $(\zeta_1-\zeta_2)^2=[(\zeta_1)^2+(\zeta_2)^2]-2\zeta_1\zeta_2$.
The $[(\zeta_1)^2+(\zeta_2)^2]$ results in a proportionality to $\mathcal I_1\rho_2^{\rm max}+\mathcal I_2\rho_1^{\rm max}$ after integrating, and hence, we obtain
\begin{equation}\label{eq:finally}
\begin{split}
    &-8\pi \alpha_A\int dr\, \frac{e^{\Lambda-3\Phi}}{r^2}(\epsilon_1+p_1)(\epsilon_2+p_2)[(\zeta_1)^2+(\zeta_2)^2]\\
    \geq&-8\pi \alpha_A(\mathcal I_1\rho_2^{\rm max}+\mathcal I_2\rho_1^{\rm max}) \ .
\end{split}
\end{equation}
The term in Eq.~\eqref{eq:beta} is effectively the same as that of Eq.~\eqref{eq:finally} but with $-\alpha_A\to \beta_A$. 

We once again obtain a term with $+2\zeta_1\zeta_2$ from expanding $(\zeta_1-\zeta_2)^2$ in Eq.~\eqref{eq:lastsquare}.
We deal with this by using that $2\zeta_1\zeta_2=(\zeta_1+\zeta_2)^2-[(\zeta_1)^2+(\zeta_2)^2]$, resulting in
\begin{equation}
\begin{split}
    &8\pi \alpha_A\int dr\, \frac{e^{\Lambda-3\Phi}}{r^2}(\epsilon_1+p_1)(\epsilon_2+p_2)2\zeta_1\zeta_2\\
    =&8\pi\alpha_A\int dr\, \frac{e^{\Lambda-3\Phi}}{r^2}(\epsilon_1+p_1)(\epsilon_2+p_2)\{(\zeta_1+\zeta_2)^2-[(\zeta_1)^2+(\zeta_2)^2]\}\\
    \geq&-8\pi \alpha_A\int dr\, \frac{e^{\Lambda-3\Phi}}{r^2}(\epsilon_1+p_1)(\epsilon_2+p_2)[(\zeta_1)^2+(\zeta_2)^2] \ .
\end{split}
\end{equation}
The inequality from the second to the third line results from the fact that $\alpha_A$ is positive definite.
Moreover, this term is exactly that of Eq.~\eqref{eq:finally}.
Then, the expression in Eq.~\eqref{eq:ll} results in 
\begin{equation}\label{eq:l_term}
    \int dr\, l\zeta_1\zeta_2\geq 4\pi (\beta_A-2\alpha_A)(\mathcal I_1\rho_2^{\rm max}+\mathcal I_1\rho_1^{\rm max}) \ .
\end{equation}

We see from Eqs.~\eqref{eq:b22} and \eqref{eq:l_term} that the bound on $a_4^0$ is proportional to $\mathcal I_1\rho_2^{\rm max}+\mathcal I_2\rho_1^{\rm max}$ and it is multiplying a negative definite term.
In fact, by considering Eq.~\eqref{eq:a1a2} for the bound on $a_4^1+a_4^2$, the full expression for $a_4$ is proportional to $\mathcal I_1\rho_2^{\rm max}+\mathcal I_2\rho_1^{\rm max}$,
with the multiplying term being negative definite.
Hence, similarly to Eq.~\eqref{eq:semicomb}, since $\mathcal I_X$ and $\rho_X^{\rm max}$ are positive definite, we have
\begin{equation}\label{eq:combine}
    -(\mathcal I_1\rho_2^{\rm max}+\mathcal I_2\rho_1^{\rm max})\geq -(\rho_1^{\rm max}+\rho_2^{\rm max})(\mathcal I_1+\mathcal I_2) \ .
\end{equation}
Then, since $a_4$ is bounded by a positive term multiplying $-(\mathcal I_1\rho_2^{\rm max}+\mathcal I_2\rho_1^{\rm max})$,
by Eq.~\eqref{eq:combine}, it is also bounded by the same term multiplying $(\rho_1^{\rm max}+\rho_2^{\rm max})(\mathcal I_1+\mathcal I_2)$.
Therefore, taking all terms into account from Eqs.~\eqref{eq:a1a2}, \eqref{eq:b22} and \eqref{eq:l_term},
we find that the bound we want for $a_4$ is
\begin{widetext}
\begin{equation}\label{eq:fulla4}
\begin{split}
    \frac{a_4}{\mathcal I_1+\mathcal I_2}\geq &-48\pi^2(\rho_1^{\rm max}+\rho_2^{\rm max})^2\left[\sup(r^2e^{4\Lambda})+\sup(r^2c_A^2e^{4\Lambda})\right]\\
    &+4\pi (\rho_1^{\rm max}+\rho_2^{\rm max})\left(\beta_A-2\alpha_A-\alpha_1-\alpha_2\right)\\
    &-4\pi(\rho_1^{\rm max}+\rho_2^{\rm max})\sup\left[c_A^2e^{2\Lambda}(1-2r\Phi')-2re^{2\Lambda}\Phi'\frac{c_A^2}{c_H^2}\right] \ .
\end{split}
\end{equation}
\end{widetext}
If we take $c_A^2\leq 1$ and $-2\alpha_A \geq -(\alpha_1+\alpha_2),2\beta_A\geq \beta_1+\beta_2$,
we find the less stringent bound given in Eq.~\eqref{eq:a4}.

To obtain the correct one-fluid limit for $a_4$, we cannot simply take the one-fluid limit on the right-hand side of Eq.~\eqref{eq:a4}.
Instead, we must use the fact that the bound on $a_4$ is directly proportional to $\mathcal I_1\rho_2^{\rm max}+\mathcal I_2\rho_1^{\rm max}$
(by Eqs.~\eqref{eq:a1a2}, \eqref{eq:b22} and \eqref{eq:l_term}).
Then, taking the one-fluid limit on the resulting expression, we use that $\mathcal I_1\to 0,\rho_1^{\rm max}\to 0$, and thus, the bound on $a_4$ in the one-fluid limit is just zero, as expected.

\bibliography{Paper/references}

\end{document}